\documentclass[unnumsec,webpdf,contemporary,large]{oup-authoring-template}%

\usepackage[mathlines, switch]{lineno}

\theoremstyle{thmstyleone}%
\newtheorem{theorem}{Theorem}
\newtheorem{proposition}[theorem]{Proposition}%
\theoremstyle{thmstyletwo}%
\theoremstyle{thmstylethree}%

\begin{document}

\journaltitle{PNAS Nexus}
\DOI{10.1093/pnasnexus/pgae103}
\copyrightyear{2024}
\pubyear{2024}
\access{Advance Access Publication Date: Day Month 2024}
\appnotes{Paper}

\firstpage{1}


\title[Scattering Spectra Models for Physics]{Scattering Spectra Models for Physics}

\author[a, $\ast$]{Sihao Cheng}
\author[b]{Rudy Morel} 
\author[c]{Erwan Allys}
\author[d]{Brice M\'enard}
\author[b,e,f]{St\'ephane Mallat}

\authormark{Cheng et al.}

\address[a]{\orgdiv{School of Natural Sciences}, \orgname{Institute for Advanced Study}, \orgaddress{\street{1 Einstein Dr, Princeton}, \postcode{08540}, \state{NJ}, \country{USA}}}
\address[b]{\orgdiv{Departement d'informatique de l'ENS}, \orgname{ENS, CNRS, PSL University}, \orgaddress{\street{Paris}, \postcode{75014}, \state{Paris}, \country{France}}}
\address[c]{\orgdiv{Laboratoire de Physique de l'Ecole normale sup\'erieure}, \orgname{ENS, Universit\'e PSL, CNRS, Sorbonne Universit\'e, Universit\'e Paris Cit\'e}, \orgaddress{\street{Paris}, \postcode{75014}, \state{Paris}, \country{France}}}
\address[d]{\orgdiv{Department of Physics and Astronomy}, \orgname{Johns Hopkins University}, \orgaddress{\street{3400 N Charles Street, Baltimore}, \postcode{21218}, \state{MD}, \country{USA}}}
\address[e]{\orgdiv{}, \orgname{Coll\`ege de France}, \orgaddress{\postcode{75231}, \state{Paris}, \country{France}}}
\address[f]{\orgdiv{Center for Computational Mathematics}, \orgname{Flatiron Institute}, \orgaddress{\street{New York}, \postcode{}, \state{NY}, \country{USA}}}

\corresp[$\ast$]{To whom correspondence should be addressed: \href{email:scheng@ias.edu}{scheng@ias.edu}}

\received{Date}{0}{2023}
\accepted{Date}{0}{2024}


\abstract{Physicists routinely need probabilistic models for a number of tasks such as parameter inference or the generation of new realizations of a field. Establishing such models for highly non-Gaussian fields is a challenge, especially when the number of samples is limited. 
In this paper, we introduce scattering spectra models for stationary fields and we show that they provide accurate and robust statistical descriptions of a wide range of fields encountered in physics.
These models are based on covariances of scattering coefficients, i.e. wavelet decomposition of a field coupled with a pointwise modulus.
After introducing useful dimension reductions taking advantage of the regularity of a field under rotation and scaling, we validate these models on various multiscale physical fields and demonstrate that they reproduce standard statistics, including spatial moments up to $4^{\text{th}}$ order. 
The scattering spectra provide us with a low-dimensional structured representation that captures key properties encountered in a wide range of physical fields. These generic models can be used for data exploration, classification, parameter inference, symmetry detection, and component separation.}
\keywords{reduced models, Gibbs energy, wavelets} 

\boxedtext{
Physicists need to characterize fields with a variety of structures, but building probabilistic models beyond the simple Gaussian model is often challenging, especially when the number of data samples is limited. We introduce scattering spectra models that make use of symmetry and regularity properties of physical fields and show that they can provide accurate and compact statistical descriptions for a wide range of fields. Providing both summary statistics and generative models, this representation can be used for data exploration, classification, parameter inference, symmetry detection, and component separation in analyzing the ever-growing datasets in physics and beyond.
}

\maketitle

\newcommand{\Cov} {{\rm Cov}}
\newcommand{\R} {{\mathbb R}}
\newcommand{\E} {{\mathbb E}}
\newcommand{\Z} {{\mathbb Z}}
\newcommand{\la} {{\lambda}}
\newcommand{\ga} {{\gamma}}
\newcommand{\sple}[1]{\bar{#1}}
\newcommand{\Av}[1] {\underset {#1} {\rm Ave}\,}
\newcommand{\Avk} {\underset {k} {\rm Ave}\,}
\newcommand{\Avu} {\underset {u} {\rm Ave}\,}
\newcommand{\Avr} {\underset {r \in G} {\rm Ave}\,}
\newcommand{\Avj} {\underset {j} {\rm Ave}\,}
\newcommand{\Avi} {\underset {i} {\rm Ave}\,}
\newcommand{\Aviu} {\underset {i,u} {\rm Ave}\,}
\newcommand{\half} {{ 1/ 2}}
\newcommand{\dd} {{\mathrm d}}



\section{Introduction}

An outstanding problem in statistics is to estimate the probability distribution $p(x)$ of high dimensional data $x$ from few or even one observed sample. In physics, establishing probabilistic models of stochastic fields is also ubiquitous, from the study of condensed matter to the Universe itself. Indeed, even if physical systems can generally be described by a set of differential equations, it is usually not possible to fully characterize their solutions. Complex physical fields, described here as non-Gaussian random processes $x$, may indeed include intermittent phenomena as well as coherent geometric structures such as vortices or filaments. Having realistic probabilistic models of such fields however allows for considerable applications, for instance to accurately characterize and compare non-linear processes, or to separate different sources and solve inverse problems. Unfortunately, no generic probabilistic model is available to describe complex physical fields such as turbulence or cosmological observations. This paper aims at providing such models for stationary fields, 
which can be estimated from one observed sample only.

At thermal equilibrium, physical systems are usually characterized by the Gibbs probability distribution, also called Boltzmann distribution, that depends on the energy of the systems~\citep{landau2013statistical}. For non-equilibrium systems, at a fixed time one may still specify the probability distribution of the field with a Gibbs energy, which is an effective Hamiltonian providing a compact representation of its statistics. Gibbs energy models can be defined as maximum entropy models conditioned by appropriate moments~\citep{jaynes_1957}. The main difficulty is to define and estimate the moments which specify these Gibbs energies.

For stationary fields, whose probability distributions are invariant to translation, moments are usually computed with a Fourier transform, which diagonalizes the covariance matrix of the field. The resulting covariance eigenvalues are the Fourier power spectrum. However, capturing non-Gaussian properties requires to go beyond second-order moments of the field. Third and fourth-order Fourier moments are called bispectrum and trispectrum. 
For a cubic $d$-dimensional stationary field of length $L$, the number of coefficients in the raw power spectrum, bispectrum and trispectrum are $O(L^d)$, $O(L^{2d})$ and $O(L^{3d})$ respectively. High-order moment estimators have high variance and are not robust, especially for non-Gaussian fields, because of potentially rare outliers which are amplified. It is thus very difficult to accurately estimate these high-order Fourier spectra from a few samples. Accurate estimations require considerably reducing the number of moments and eliminating the amplification effect of high-order moments.

Local conservation laws for mass, energy, momentum, charge, etc. result in continuity equations or transport equations. The resulting probability distributions of the underlying processes thus are typically regular to deformations that approximate the local transport. These properties have motivated many researchers to make use of a wavelet transform as opposed to a Fourier transform, which provides localized descriptors. Most statistical studies have concentrated on second-order and marginal wavelet moments \citep[e.g.,][]{Bougeret_1995, Vielva_2004, Podesta_2009} which fail to capture important non-Gaussian properties of a field. Other studies~\citep[][]{ha2021adaptive} use wavelet operator for interpretation with application to cosmological parameter inference, but rely on a trained neural network model. 

In recent years, new representations have been constructed by applying point-wise non-linear operators on the wavelet transforms to recover their non-Gaussian information. The scattering transform, for instance, is a representation that is built by cascading wavelet transforms and non-linear modulus~\citep{mallatscat, bruna2013invariant}. This representation has been used in astrophysics and cosmology to study the interstellar medium~\citep{allys2019rwst,saydjari2021classification}, weak-lensing fields \citep{Cheng_2020,Cheng_2021}, galaxy surveys~\citep{valogiannis2022towards, Valogiannis_2022, Valogiannis_2023}, and radio observations~\citep{Greig_2022} (readers with physics background may find this review \cite{cheng2021quantify} useful). Other representations, which are built from covariances of phase harmonics of wavelet transforms~\citep{mallat2020phase,zhang2021maximum}, have also been used to model different astrophysical processes~\citep{Allys_2020, jeffrey2022single, regaldo2023generative}. Such models, which can be built from a single image, have in turn enabled the development of new component separation methods~\citep{regaldo2021new, delouis2022non}, which can be directly applied to observational data without any particular prior model of the components of a mixture~\citep{auclair2023separation,siahkoohi2023unearthing}.

These models however suffer from a number of limitations: they are not very good at reproducing vortices or long thin filaments, and they require an important number of coefficients to capture dependencies between distant scales, as well as angular dependencies. Building on those previous works, \cite{rudy} introduced a reduced scattering spectra representations for time series by leveraging scale invariance. In this paper, we present the scattering spectra for datasets with dimensions more than one, which is a low-dimensional representation able to efficiently describe a wide range of non-Gaussian processes encountered in physics. In particular, we show that it is possible to take into account the intrinsic regularity of physical fields and dramatically reduce the size of such representations. 
The first part of the paper presents maximum entropy models and the scattering spectra statistics, as well as their dimensional reduction. The second part of the paper presents a quantitative validation of these models on various two-dimensional multiscale physical fields and discuss their limitations. 

\begin{figure*}
    \centering
    \includegraphics[width=\textwidth]{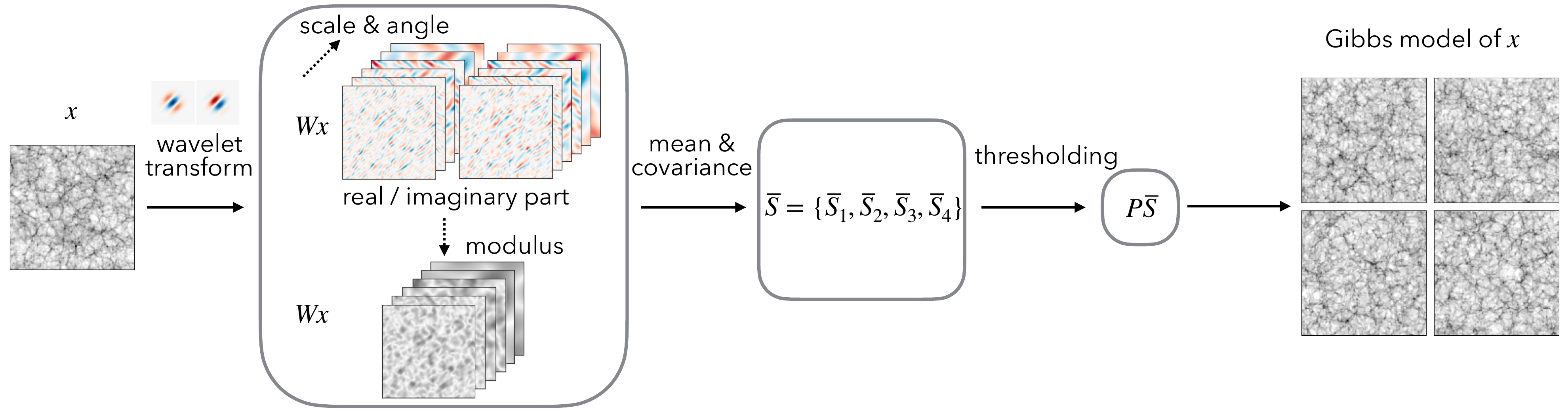}
    \caption{Steps to build a feasible model for a random field $x$ from only one or a few realizations. We first build a low-dimension representation $\Phi(x)$ of the random field, which specifies a maximum entropy model. The representation $\Phi(x)$ is obtained by conducting the wavelet transform $Wx$ and its modulus $|Wx|$, and then computing the means and covariance of all wavelet channels $(Wx\,,\,|Wx|)$. Such a covariance matrix is further binned and sampled using wavelets to reduce its dimensionality, which is called the scattering spectra $\bar{S}(x)$. Finally, These scattering spectra are renormalized and reduced in dimension by thresholding its Fourier coefficients along rotation and scale parameters $\Phi(x)=P\bar{S}$, making use of the regularity properties of the field. For many physical fields, this representation can be as small as only around $~\sim 10^2$ coefficients for a 256$\times$256 field.}
    \label{fig:model-concept}
\end{figure*}

{\bf Notations:}  $v^*$ is the complex conjugate of a scalar $v$.  ${\rm Ave}_i$ averages values indexed by $i$ in a finite set. $\hat x[k]$ is the Fourier transform of $x[u]$, whether $u$ is a continuous variable in $\R^d$ or belongs to finite periodic lattice. $\E\{ \Phi(x)\}$ is the expectation of $\Phi(x)$ according to the probability distribution $p(x)$ of a vector $x$. $\log$ stands for base 2 logarithm.


\section{Methods}
\subsection{Gibbs Energy of Stationary Fields}
\label{expon-sec}

We first review the properties of Gibbs energies resulting from maximum entropy models conditioned by moment values \citep{German_1984, zhu1997minimax, zhu1998filters}. We write $x[u]$ a field where the site index $u$ belongs to a cubic $d$-dimensional lattice of size $L$. It results that $x \in \R^{L^d}$.

Assume that $x \in \R^{L^d}$ has a probability density $p(x)$ and consider Gibbs energy models linearly parameterized by a vector $\theta = \{\theta_m \}_{m \leq M}$ over a potential vector $\Phi(x) = \{ \Phi_m (x) \}_{m \leq M}$ of dimension $M$ 
\begin{align}
U_\theta (x) = \langle \theta , \Phi(x) \rangle = \sum_{m=1}^M \theta_m^*\, \Phi_m(x).
\end{align}
They define exponential probability models
\begin{equation}
\label{exponmod}
    p_\theta (x) = Z^{-1}_\theta \, e^{- \langle \theta , \Phi(x) \rangle}.
\end{equation}
The model class is thus defined by the potential vector $\Phi(x)$, which needs to be chosen appropriately. 

If it exists, the maximum entropy distribution conditioned by $\E\{\Phi(x)\}$ is a $p_{\theta_0}$ which belongs to this model class. It has a maximum entropy $H(p_{\theta_0}) = - \int p_{\theta_0} (x)\, \log p_{\theta_0} (x)\,dx$ under the expected value condition
\begin{equation}
\label{moment-cond2}
    \int \Phi(x)\, p_{\theta_0} (x)\, dx = \E\{\Phi(x)\} .
\end{equation}
In statistical physics, $p_{\theta_0}$ is a macrocanonical model defined by a vector $\E\{\Phi(x)\}$ of observables. One can verify that $\theta_0$ also minimizes the Kullback-Liebler divergence within the class
\begin{equation}
    \label{model-error}
D(p \| p_{\theta_0}) = \int p(x) \log \frac{p(x)} {p_{\theta_0} (x)}\, dx = H(p_{\theta_0}) - H(p) .
\end{equation}

The main topic of the paper is to specify $\Phi(x)$ in order to define accurate maximum entropy models for large classes of physical fields, which can be estimated from a small number $n$ of samples $\sple{x}_i$. In this section, we suppose that $n=1$. Reducing the model error given by \eqref{model-error} amounts to defining $\Phi$ which reduces the excess entropy of the model. This can be done by enriching $\Phi(x)$ and building very high-dimensional models. However, we must also take into account the empirical estimation error of $\E \{ \Phi(x)\}$ by $\Phi(\sple{x}_1)$, measured by $\E\{\| \Phi(x) -  \E\{\Phi(x)\}\|^2\}$. 

In this paper, macrocanonical models are approximated by microcanonical models, which have a maximum entropy over a microcanonical set of width $\epsilon > 0$
\begin{equation}
    \label{micro-set}
    \Omega_\epsilon = \{ x \in \R^{L^d}~:~\|\Phi(x) -  \Phi(\sple{x}_1) \|^2 \leq \epsilon \} .
\end{equation}
Supplementary material A reviews a sampling algorithm for such model. It also explains how to extend the definition of $\Omega_\epsilon$ for $n > 1$ samples $\sple{x}_i$ by replacing $\Phi(\sple{x}_1)$ by ${\rm Ave}_i \Phi(\sple{x}_i)$.
If $\Phi(x)$ concentrates around $\E\{\Phi(x)\}$ then the microcanonical model converges to the macrocanonical model when the system length $L$ goes to $\infty$ and $\epsilon$ goes to $0$. The concentration of $\Phi(x)$ generally imposes that its dimension $M$ is small relatively to the dimension $L^d$ of $x$. The choice of $\Phi(x)$ must thus incorporate a trade-off between the model error \eqref{model-error} and the distance between micro and macrocanonical distributions.

\subsection{Fourier Polyspectra} 

Gaussian random fields are maximum entropy models conditioned on first and second-order moments. The potential vector $\Phi(x)$ is then an empirical estimator of first and second-order moments of $x$. For stationary fields, there is only one first-order moment $\E\{x[u]\}$ which can be estimated with an empirical average\footnote{This single moment can be directly constrained, and we do not discuss it in the following.} over $u$: $\text{Ave}_u x[u]$. Similarly, the covariance matrix $\E\{x[u] x[u']\}$ only depends on $u-u'$, so only the diagonal coefficients in Fourier space are informative, which are called the power spectrum,
\begin{equation}
\E\{ \hat x [k]\,\hat x[k']^* \} ~~\mbox{with}~~ k = k'.
\label{eq:power-spectrum}
\end{equation}
The off-diagonal elements vanish because of phase cancellation under all possible translations, which means the second-order moments treat Fourier coefficients independently, and cannot describe relations or dependence between them. The diagonal elements, which can also be written as $|\hat x [k]|^2$, can be estimated from a single sample $x$ by averaging $|\hat x[k]|^2$ over frequency bins that are large enough to reduce the estimator variance. A uniform binning and sampling along frequencies results in power spectrum estimators with $O(L^d)$ elements, so the Gaussian model is compact and feasible.

However, the Gaussian random field model has limited power to describe complex structures. The majority of fields encountered in scientific research are not Gaussian. Non-Gaussianity usually means dependence between Fourier coefficients at different frequencies.
The traditional way goes to higher orders moments of $\hat x$, the polyspectra \citep{brillinger1965introduction}, where phase cancellation implies that for stationary fields, only the following moments are informative,
\begin{equation}
\E\{\hat x[k_1]\,...\hat x[k_n] \} ~~\mbox{with}~~k_1 + ... + k_n = 0 ,
\label{eq:polyspectra}
\end{equation}
while other moments are zero. These polyspectra at order $n>2$ capture dependence between $n-1$ independent frequencies. As the leading term, the Fourier bispectrum specifies the non-zero third-order moments and has $O(L^{2d})$ coefficients. However, bispectrum is usually not sufficient to characterize non-Gaussian fields. For example, it vanishes if the field distribution is symmetric $p(x) = p(-x).$ One must then estimate fourth-order Fourier moments, the trispectrum, which has $O(L^{3d})$ coefficients. 

There are two main problems for the polyspectra coefficients to become proper potential functions $\Phi(x)$ in the maximum entropy models. First, the number of coefficients increases sharply with the order. Second, high-order moments are not robust and difficult to estimate from a few realizations \citep{Huber_1981}. For random fields with a heavy tail distribution, which is ubiquitous in complex systems \citep{Bak_1987, Bouchaud_1990, Coles_1991, Kello_2010, Sornette_2017}, higher order moments may not even exist. Those two problems are common for high-order moments and have been demonstrated in real-world applications \citep{Dudok_2004, Lombardo_2014}. In the following two sections, we introduce modifications to this approach to solve those problems.

\subsection{Wavelet Polyspectra}
\label{reduscatcov}

Many physical fields exhibit multiscale structures induced by non-linear dynamics, which implies regularity of $p(x)$ in frequency. The wavelet transform groups Fourier frequencies by wide logarithmic bands, providing a natural way to compress the Fourier polyspectra. The compression not only reduces the model size but also improves estimator convergence. We use the wavelet transform to compute a compressed power spectrum estimate, as well as a reduced set of $O(\log^2 L)$ third and $O(\log^3 L)$ fourth order wavelet moments, allowing for efficient estimation of the polyspectra.

\subsubsection{Wavelet Transform}
A wavelet is a localized wave-form $\psi[u]$ for $u \in \R^d$ which has a zero average $\int_{\R^d} \psi[u]\,du = 0$. We shall define complex-valued wavelets $\psi[u] = g[u]\,e^{i \xi . u}$ where $g[u]$ is a real window whose Fourier transform $\hat g[k]$ is centered at $k=0$ so that $\hat \psi[k] = \hat g[k-\xi]$ is localized in the neighborhood of the frequency $\xi$. Fig. S1 shows $\psi$ and $\hat \psi$ for a $d=2$ dimensional Morlet wavelet described in supplementary material B. The wavelet transform is defined by rotating $\psi[u]$ with a rotation $r$ in $\R^d$ and by
dilating it with dyadic scales $2^j > 1$. It defines
\begin{equation}
\label{wave}
    \psi_\la [u] = 2^{-j d}\, \psi[2^{-j} r^{-1} u]~~\mbox{with}~~\lambda = 2^{-j}\, r \xi~.
\end{equation}
Its Fourier transform is $\hat \psi_\la [k] = \hat g[2^j r^{-1} (k - \xi)]$, which is centered at the frequency $\lambda$ and concentrated in a ball whose radius is proportional to $2^{-j}$.

To decompose a field $x[u]$ defined over a grid of width $L$, the wavelet is sampled on this grid. Wavelet coefficients are calculated as convolutions with periodic boundary conditions
\begin{equation}
\label{conv_prod}
    W x[u,\lambda] = x \star \psi_{\lambda} [u] = \sum_{u'} x[u']\, \psi_{\lambda}[u-u'] .
\end{equation}
It measures the variations of $x$ in a spatial neighborhood of $u$ of length proportional to $2^j$, and it depends upon the values of $\hat x$ in a frequency neighborhood of $k=\lambda$ of length proportional to $2^{-j}$. The scale $2^j$ is limited to $1 \leq j \leq J$, and for practical application to fields with a finite size $L$, the choice of $J$ is limited by $J < \log L$. Left part of Fig.~\ref{fig:model-concept} illustrates the wavelet transform of an image.

The rotation $r$ is chosen within a rotation group of cardinal $R$, where $R$ does not depend on $L$. Wavelet coefficients need to be calculated for $R/2$ rotations because $Wx[u,-\lambda] = Wx[u,\lambda]^*$ for real fields. In $d = 2$ dimensions, the $R$ rotations have an angle $2 \pi \ell / R$, and we set $R = 8$ in all our numerical applications, which boils down to 4 different wavelet orientations. The total number of wavelet frequencies $\lambda$ is $R J = O(\log L)$\footnote{Here we assume the choice of $R$ is independent of field dimension $d$. Another possible choice is to require a constant ratio between the radial and tangential sizes of the $d$-dimension oriented wavelets. Then, $R$ is  proportional to the ratio between the surface area of a $d$--1-sphere and the volume of a $d$--1-ball, proportionally to $\Gamma(n/2+1/2)/\Gamma(n/2)$. It results in an approximate scaling of $R J = O(d\log L)$ when $d$ is small and $O(\sqrt{d}\log L)$ when $d$ is large.} as opposed to $L^d$ Fourier frequencies.

A wavelet transform is also stable and invertible if $\psi$ satisfies a Littlewood-Paley condition, which requires an additional convolution with a low-pass {\it scaling} function $\psi_0$ centered at the frequency $\lambda = 0$. The specifications are detailed in supplementary material B. 

\subsubsection{Wavelet Power Spectrum} 

Given scaling regularity, one can compress the $O(L^d)$ power spectrum coefficients into $RJ=O(\log L)$ coefficients using a logarithmic binning defined by wavelets. This is obtained by averaging the power spectrum with weight functions as the Fourier transform of wavelets, which are band-pass windows, $\text{Ave}_k \left( \E\{|\hat x[k]|^2\}\, |\hat \psi_{\lambda}[k]|^2 \right)$. The limited number of wavelet power spectrum coefficients has reduced estimation variance. In fact, they are also the diagonal elements of the wavelet covariance matrix, $W x[u,\lambda]W x[u,\lambda]^*=|W x[u,\lambda]|^2$, therefore an empirical estimation can also be written as an average over $u$:
\begin{equation}
    \label{spec-wave}
    M_2 = \Avu |W x[u,\lambda ]|^2.
\end{equation}
Similar to the power spectrum, phase cancellation due to translation invariance means that the off-diagonal blocks i.e. the cross-correlations between different wavelet frequency bands are nearly zero because the support of two wavelets $\hat \psi_\la$ and $\hat \psi_{\la'}$ are almost disjoint, as illustrated in Fig.~\ref{fig:FourierSupportalignment}(a). 

\subsubsection{Selected 3rd and 4th Order Wavelet Moments}

One may expect to compress the polyspectra in a similar manner with a wavelet transform, taking advantage of the regularities of the field probability distribution. However, it is non-trivial to logarithmically bin the polyspectra because more than one independent frequency is involved and the phase cancellation condition needs to be considered. 

To solve this problem, let us revisit the phase cancellation of two frequency bands, which causes their correlation to be zero, 
\begin{align}
    \E\{W x[u,\la]\,W x[u',\la']^*\} \sim 0\, ,
\end{align}
for $\la \neq \la'$. To create a non-zero correlation, we must realign the support of $W x[u,\la]$ and $W x[u',\la']$ in Fourier space through non-linear transforms. As shown in Fig.~\ref{fig:FourierSupportalignment}(b), we may apply a square modulus to one band (shown in blue) in the spatial domain, which recenters its frequency support at origin. Indeed, $|x \star \psi_\la |^2 = (x \star \psi_\la )(x \star \psi_\la )^*$ has a Fourier support twice as wide as that of $x \star \psi_\la$, and will overlap with another wavelet band with lower frequency than $\la$. The transformed fields $|x \star \psi_\la |^2$ can be interpreted as maps of locally measured power spectra. Correlating this map with another wavelet band $x \star \psi_\la'$ gives some third-order moments 
\begin{align}
    \E\{ |Wx|^2 [u,\la]\, Wx [u',\la']^* \}
\end{align}
that are a priori non-zero. 
Furthermore, for wide classes of multiscale processes having a regular power spectrum, it suffices to only keep the coefficients at $u=u'$ because of random phase fluctuation (see supplementary material~B). For stationary random fields, they can be estimated with an empirical average over $u$,
\begin{equation}
    \label{thirwavemo}
M_3 = \Avu( |Wx|^2[u,\la]\, Wx [u,\la']^* ).
\end{equation}
Now we obtain a set of statistics characterizing the dependence of Fourier coefficients in two wavelet bands in a collective way, which are selected third-order moments. They can be interpreted as a logarithmic frequency binning of certain bispectrum coefficients. There are about $R^2 J^2 = O(\log^2 L)$ such coefficients, which is a substantial compression compared to the $O(L^{2d})$ full bispectrum coefficients. 

%
\begin{figure}[!t]
\centering
\includegraphics[width=\columnwidth]{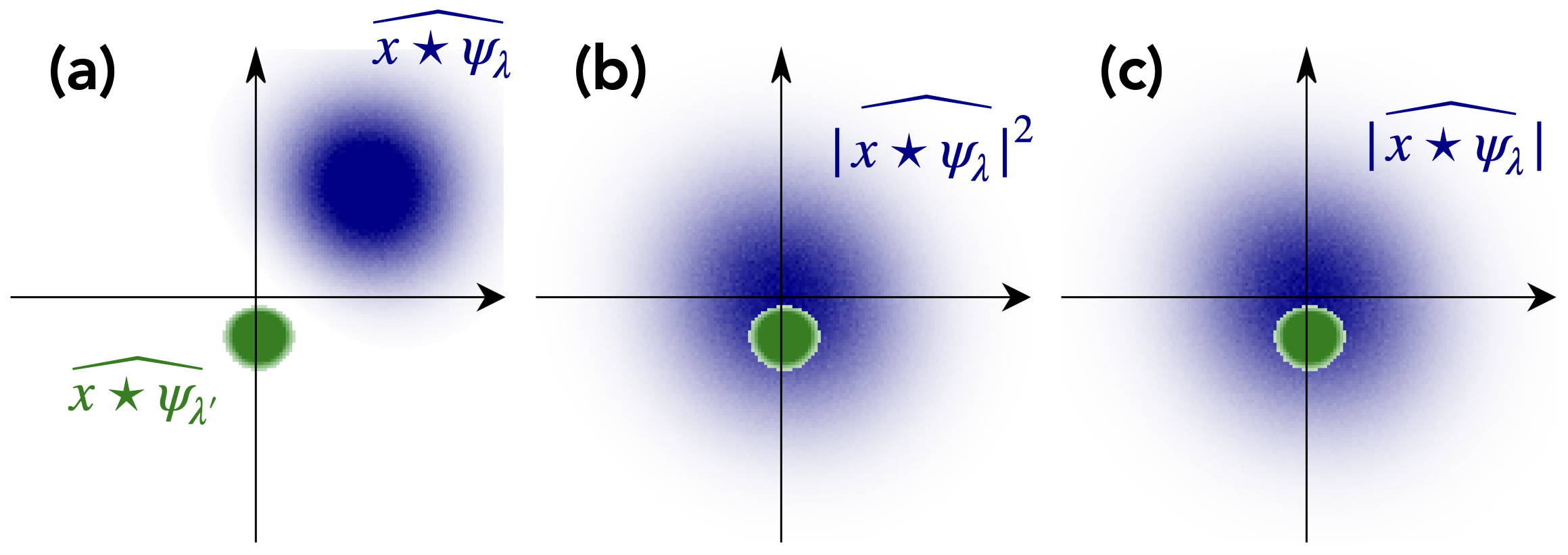}
\caption{(a): For $\lambda\ne\lambda'$ the Fourier supports of $x\star\psi_{\lambda}$ (blue) and $x\star\psi_{\lambda'}$ (green) typically do not overlap. (b): The Fourier support of $|x\star\psi_\lambda|^2$ is twice larger and centered at $0$ and hence overlaps with $x\star\psi_{\lambda'}$ if $|\la'| \leq |\la|$. (c): The Fourier support of $|x\star\psi_\lambda|$ is also centered at $0$ and hence overlaps with $x\star\psi_{\lambda'}$ if $|\la'| < |\la|$. 
}
\label{fig:FourierSupportalignment}
\end{figure}

Similarly, we consider the cross correlation between two wavelet bands both transformed by the square modulus operation and obtain a wavelet binning of fourth-order moments,
\begin{align}
    \E\{ |Wx [u,\la]|^2 \, |Wx [u',\la']|^2 \} - \E\{ |Wx [u,\la]|^2 \} \E\{ |Wx [u',\la']|^2 \}.
\end{align}
For stationary fields, this covariance only depends on $u-u'$. 
A further reduction of such a large covariance function is possible because its Fourier transform over $u-u'$ has two properties. First, it typically does not have higher frequency components than the initial wavelet transforms involved (see Fig.~\ref{fig:FourierSupportalignment}) as the phase fluctuations have been eliminated by the square modulus, and second, for fields with multiscale structures, it is regular and can be approximated with another logarithmic frequency binning. Thus, we can compress the large covariance function with a second wavelet transform, and estimate it by an empirical average over $u$:
\begin{equation}
    \label{fourthorder0}
    M_4 = \Avu ( W|Wx|^2[u,\la,\gamma]~~W|Wx|^2[u,\la',\gamma]^* ),
\end{equation}
where $(W|Wx|^2) [u,\lambda,\gamma] = |x \star \psi_{\lambda}|^2 \star \psi_{\gamma} [u]$, and the central frequencies of the second wavelets verifies $|\lambda|\geq |\lambda'|>|\gamma|$. There are about $R^3 J^3 = O(\log^3 L)$ such coefficients, which is also a substantial compression compared to the $O(L^{3d})$ full trispectrum coefficients. 

\subsection{Scattering Spectra}
\label{scat-spec-sec}

In general, the estimation of high-order moments has a high variance because high-order polynomials amplify the effect of outliers. An interesting idea learned from the scattering transform approach \citep{mallatscat, bruna2013invariant} is that the multiplication $(Wx)(W^{*}x)=|Wx|^2$ used in the higher-order moments~\eqref{spec-wave}, \eqref{thirwavemo}, \eqref{fourthorder0} can be replaced by the wavelet modulus $|Wx|$, which produces qualitatively similar estimators but with improved robustness and better efficiency in presence of sparse structures \footnote{More than a hundred years ago, astronomer Sir Arthur Eddington expressed in his book \citep[p.147]{Eddington_1914} a favor of mean-modulus error over mean-square error for similar reasons.}\citep{cheng2021quantify}.
The resulting moments after such a replacement only depend on the mean and covariance matrix of $(Wx, |Wx|)$, which are low-order transforms of the original field $x$. 

Local statistics of wavelet modulus have been studied to analyze properties of image textures \citep{portilla2000parametric}. Their mathematical properties have been analyzed to capture non-Gaussian characteristics of random fields \citep{mallat2020phase, zhang2021maximum} in relation to scattering moments \citep{mallatscat,bruna2013invariant}. Scattering spectra have been defined on one-dimensional time-series \citep{rudy}, from the joint covariance of a wavelet transform and its modulus: $(Wx\,,\,|Wx|)$. We extend it to fields of arbitrary dimension $d$ and length $L$, in relation to Fourier high-order moments, and define models of dimension $O(\log^3 L)$.

\subsubsection{First and second wavelet moments, sparsity}
For non-Gaussian fields $x$, wavelet coefficients $W x[u,\lambda]$ define fields which are often sparse \citep{olshausen_1996,stephane1999wavelet}.
This is a non-Gaussian property that can be captured by first-order wavelet moments $\E\{|W x[u,\lambda ]|\}$. If $x$ is a Gaussian random field then $W x[u,\lambda]$ remains Gaussian but complex-valued so, and we have $\frac{\E\{|W x|\}^2} {\E\{|W x|^2\}} = {\frac \pi {4}}$. This ratio decreases when the sparsity of $W x[u,\lambda ]$ increases. The expected value of $|Wx|$ is estimated by
\begin{equation}
\label{spec1}
    S_1(x)[\lambda] = {\Avu|W x[u,\lambda ]|} 
\end{equation}
and the ratio is calculated with the second-order wavelet spectrum estimator
\begin{equation}
\label{spec2}
S_2(x)[\la] = M_2(x)[\la] = \Avu (|W x|^2[u,\la]) .
\end{equation}

\subsubsection{Cross-Spectra between Scattering Channels}

Let us now replace $|Wx|^2$ by $|Wx|$ in the selected third and fourth-order wavelet moments described in the previous section. The third order moments \eqref{thirwavemo} become $\E\{ |Wx [u,\la]|\, Wx [u,\la']^* \}.$
Such moments are a priori non-zero if the Fourier transforms of $|Wx[u,\la]| = |x \star \psi_\la|$ and $Wx[u,\la'] = x \star \psi_{\la'}$ overlap. This is the case if $|\la'| < |\la|$ as illustrated in Fig.~\ref{fig:FourierSupportalignment}. Eliminating the square thus preserves non-zero moments which can capture dependencies between different frequencies $\la$ and $\la'$. The third order moment estimators given by \eqref{thirwavemo} can thus be replaced by lower cross-correlations between $|Wx|$ and $Wx$ at $|\la| \geq |\la'|$
\begin{equation}
\label{spec3}
    S_3(x)[\lambda,\lambda'] = {\Avu\,(|W x|[u,\lambda]\,\, W x[u,\lambda']^*)}.
\end{equation}

Replacing $|Wx|^2$ by $|Wx|$ in the fourth order wavelet moments \eqref{fourthorder0}
amounts to estimating the covariance matrix of wavelet modulus fields $|W x|$. As the $u-u'$ dependency of this covariance can also be characterized by a second wavelet transform, this amounts in turn to estimate the covariance of scattering transforms $W|Wx| [u,\lambda,\gamma] = |x \star \psi_{\lambda}| \star \psi_{\gamma} [u]$
\begin{equation}
\label{spec4}
    S_4(x) [\lambda,\lambda',\gamma] = \Avu \Big( W|Wx|[u,\lambda,\gamma]~~ W|Wx|[u,\lambda',\gamma]^*\Big) ,
\end{equation}
for $|\la| \geq |\la'| \geq |\gamma|$. It provides a wavelet spectral estimation of the covariance of $|W x|.$

Combining the moment estimators of Eqs.~(\ref{spec1},\ref{spec2},\ref{spec3},\ref{spec4}) defines a vector of scattering spectra  
\begin{equation}
\label{scatcorr}
    S (x) = \Big(S_1(x)\,,\,S_2(x)\,,\, S_3(x)\,,\, S_4(x) \Big) .
\end{equation}
It provides a mean and covariance estimation of the joint wavelet and wavelet modulus vectors $(W x , |W x|)$. It resembles the second, third, and fourth-order Fourier spectra but has much fewer coefficients and better information concentration. Considering the conditions satisfied by $\la$, $\la'$, and $\ga$, the exact dimension of $S(x)$ is $ RJ + R^2J(J-1)/8 + R^3 J(J^2-1)/48$, of the order $O(log^3 L)$.

\subsubsection{Renormalization} 

Scattering spectra coefficients must often be renormalized to improve the sampling of maximum entropy models. Indeed, multiscale random processes often have a power spectrum that has a power law decay $\E\{|\hat x[k]|^2\} \sim |k|^{-\eta}$ over a wide range of frequencies, long-range correlations corresponding to a strong decay from large to small scales. The wavelet spectrum also has a power-law decay $\E\{|Wx[u,\la]|^2 \} \sim |\la|^{-\eta}$. This means that if we build a maximum entropy model with $\Phi(x) = S(x)$ then the coordinate of $\Phi(x)$ of low-frequencies $\la$ have a much larger amplitude and variance than at high frequencies. The microcanonical model is then dominated by low frequencies and is unable to constrain high-frequency moments. The same issue appears when computing the $\theta_0$ parameters of a macrocanonical model defined in \eqref{moment-cond2}, for which it has been shown that renormalizing to $1$ the variance of wavelet coefficients at all scales avoid numerical instabilities~\citep{WCRG}. Without such a normalization, the calculation of $\theta_0$ parameters at different frequencies is ill-conditioned, which turns into a ``critical slowing down" of iterative optimization algorithms. The proposed normalization is closely related to Wilson renormalization.

We renormalize the scattering spectra by the variance of wavelet coefficients, $\sigma^2[\la]={\rm Ave}_i S_2(\sple{x_i})[\la]$, which can be estimated from a few samples.
The renormalized Scattering Spectra are
\begin{align}
    \bar S (x) = \Big(\bar S_1(x)\,,\,\bar S_2(x)\,,\, 
    \bar S_3(x)\,,\, \bar S_4(x) \Big)
\end{align}
defined by 
\begin{equation}
\label{eq:scat-spectra-normalized}
    \bar S_1(x)[\la] = \frac{S_1(x)[\la]} {\sigma[\la]}~~,~~\bar S_2(x)[\la] = \frac{S_2(x)[\la]}{\sigma^2[\la]}
\end{equation}
\begin{equation*}
    \bar S_3(x)[\la,\la'] = \frac{S_3(x)[\la,\la']} {\sigma[\la]\,\sigma[\la']}~, 
    ~\bar S_4(x)[\la,\la',\ga] = \frac{S_4(x)[\la,\la',\ga]} {\sigma[\la]\,\sigma[\la']}.
\end{equation*}
The microcanonical models proposed in this paper are built from these renormalized statistics and/or their reduced version described below.

\subsection{Dimensionality reduction for physical fields}

Though much smaller than the polyspectra representation, the scattering spectra $\bar{S}$ representation still has a large size. Assuming isotropy and scale invariance of the field $x$, a first-dimensional reduction can be performed that relies on the equivariance properties of scattering spectra with respect to rotation and scaling (see supplementary material~C). However, such invariances cannot be assumed in general. In this section, we propose to construct a low-dimensional representation by only assuming regularity under rotation or scaling of the scales involved in the scattering spectra representation. A simplified version of such a dimensional reduction has been introduced in~\cite{allys2019rwst}.
We refer the reader to supplementary material~D for technical details.

The goal of the reduction is to approximate the covariance coefficients $\bar{S}_3$ and $\bar{S}_4$, the most numerous, using only a few coefficients. This can be seen as a covariance matrix estimation problem.
To do so, we first use a linear transform to sparsify the covariance matrix and then perform a threshold clipping on the coefficients to reduce the representation.
We consider a linear transform $F\bar S = (\bar S_1, \bar S_2, F\bar S_3, F\bar S_4)$ with a pre-determined linear transform $F$ which stands for a 2D or 3D Fourier transform along all orientations, as well as a 1D cosine transform along scales, for $\bar S_3$ and $\bar S_4$. For fields with statistical isotropy or self-similarity, all harmonics related to the action of global rotation and scaling on the field $x$ should be consistent with zero, except for the zeroth harmonic. For general physical fields, we expect the statistics $\bar S(x)$ to have regular variations to the action of rotation or scaling of the different scales involved in its computation, which implies that its Fourier harmonics $F\bar S(x)$ have a fast decay away from the 0-th harmonic and $F\bar S(x)$ is a sparse representation.

Thresholding on a sparse representation is widely used in image processing for compression \citep[][]{chang2000adaptive}.
We use threshold clipping on the sparse representation $F\bar S$ to significantly reduce the size of the scattering spectra. Furthermore, when empirically estimating large but sparse covariance matrices such as $F\overline{S}$, thresholding provides Stein estimators \citep{Stein_1956} which have lower variance and are consistent\citep[e.g.,][]{Donoho_1994, Bickel_2008, Cai_2011, Fan_2013}.
As $\bar S_1$ or $\bar S_2$ are already small, we keep all of their coefficients.

There are different strategies available to set the threshold for clipping. We adopt a simple strategy which keeps those coefficients with $\mu(F\bar{S}) > 2\sigma(F\bar{S})$, where $\mu(F\bar{S})$ and $\sigma(F\bar{S})$ are the means and standard deviations of individual coefficients of $F\bar{S}$. These adaptive thresholding estimators achieve a higher rate of convergence and are easy to implement \citep{Cai_2011}. With multiple realizations from simulations, $\mu(F\bar{S})$ and $\sigma(F\bar{S})$ can be estimated directly. In the case where only a single sample field is available, $\sigma(F\bar{S})$ can be estimated from different patches of that sample field \citep[e.g.,][]{Sherman_1996}. We call $P\bar S$ the coefficients after thresholding projection:
\begin{equation}
    \label{PS}
    P\bar S = (\bar S_1, \bar S_2, P\bar S_3, P\bar S_4) = \text{thresholding } F\bar S.
\end{equation}
The compact yet informative set of scattering spectra $P \bar S$ is the representation $\Phi(x)=P\bar S(x)$ proposed in this paper to construct maximum entropy models.


\begin{figure*}
    \centering
    \includegraphics[width=0.90\textwidth]{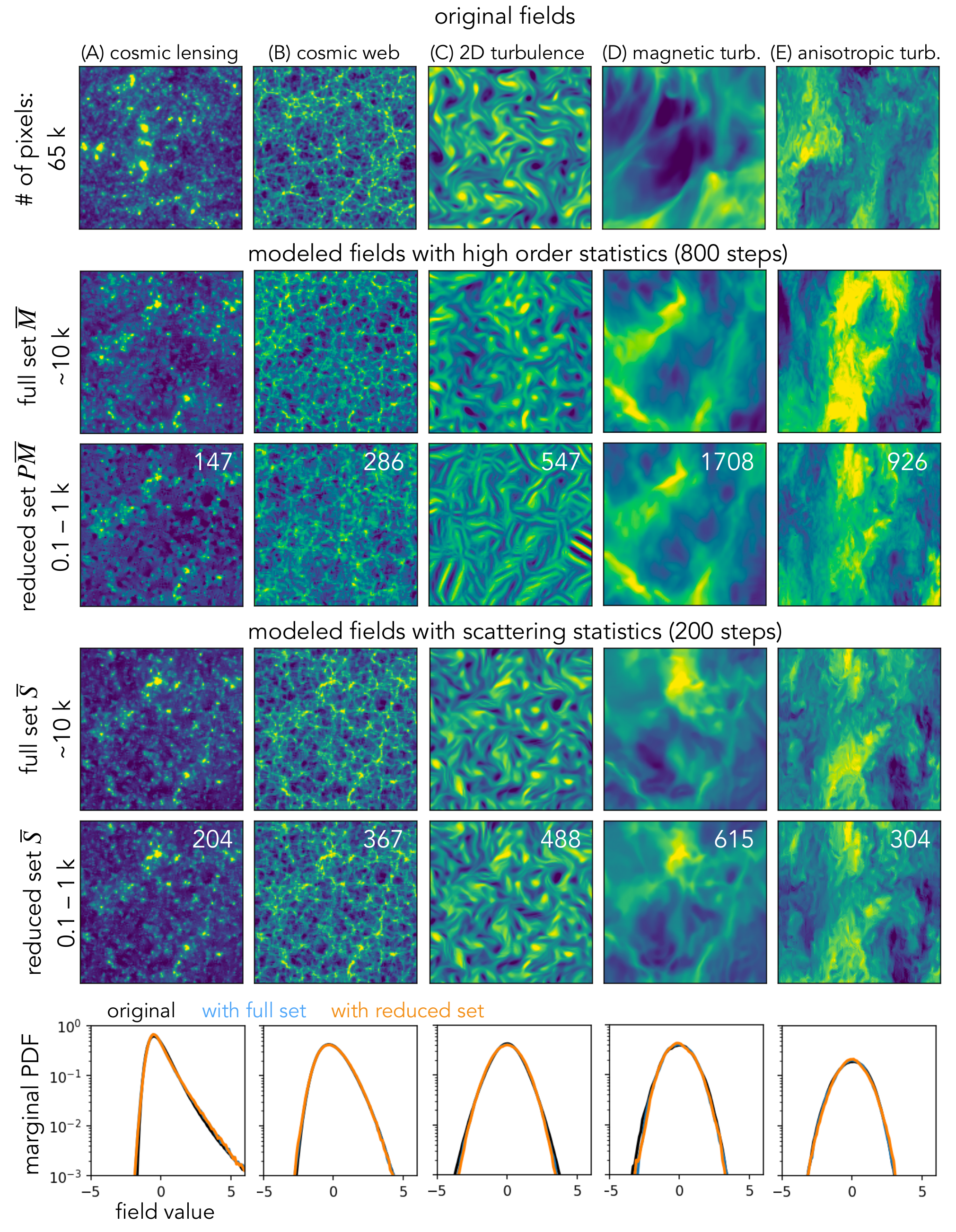}
    \caption{Visual comparison of realistic physical fields and those sampled from maximum entropy models based on wavelet higher-order moments $\bar M$ and wavelet scattering spectra $\bar S$ statistics. The first row shows five example fields from physical simulations of cosmic lensing, cosmic web, 2D turbulence, magnetic turbulence, and squeezed turbulence. The second and third rows show syntheses based on the selected high-order wavelet statistics estimated from 100 realizations. They are obtained from a microcanonical sampling with 800 steps. The fourth and fifth rows show similar syntheses based on the scattering spectra statistics, with only 200 steps of the sampling run. This figure shows visually that the scattering spectra can model well the statistical properties of morphology in many physical fields, while the high-order statistics either fail to do so or converge at a much slower rate. To clearly show the morphology of structures at small scales, we show a zoom-in of 128 by 128 pixels regions. Finally, to quantitatively validate the goodness of the scattering model, we show the marginal PDF (histogram) comparison in the last row.
    }
    \label{fig:synthesis}
\end{figure*}

\section{Numerical Results}
\label{numerical}

We have introduced maximum entropy models based on small subsets of $O(\log^3 L)$ scattering spectra moments $\bar{S}$ and projected moments $P \bar{S}$, claiming that it can provide accurate models of large classes of multiscale physical fields, and reproduce $O(L^{3d})$ power spectrum, bispectrum and trispectrum Fourier moments. This section provides a numerical justification of this claim with five types of 2D physical fields from realistic simulations. 
In order to reduce the variance of the validation statistics, we consider in this section a model estimated on several realizations of a field. However, our model also produces convincing realizations when estimated on a single realization (see Fig.~S2 for a visual assessment).
All computations are reproducible with the software available on \url{https://github.com/SihaoCheng/scattering_transform}.

\subsection{Dataset of Physical Fields}

We use five two-dimensional physical fields to test the maximum entropy models. 
The five fields are chosen to cover a range of properties in terms of scale dependence, anisotropy, sparsity, and morphology:
\begin{itemize}
    \item[(A)] {\it Cosmic lensing :} simulated convergence maps of gravitational lensing effects induced by the cosmological matter density fluctuations
    \citep{matilla2016dark, gupta2018non}. 
    \item[(B)] {\it Dark matter:} logarithm of 2D slices of the 3D large-scale distribution of dark matter in the Universe~\citep{quijote2020}.
    The logarithm allows for the filamentary cosmic web structures to stand out and thus increases the morphological diversity of our examples, which we discuss more in supplementary material~G.
    \item[(C)] {\it 2D turbulence:} turbulence vorticity fields of incompressible fluid stirred at the scale around 32 pixels, simulated from 2D Navier-Stokes equations~\citep{schneider2006coherent}.
    \item[(D)] {\it Magnetic turbulence:} column density of 3D isothermal magnetic-hydrodynamic (MHD) turbulent simulations~\citep{allys2019rwst}.
    The field is anisotropic due to a mean magnetic field in the horizontal direction. 
    \item[(E)] {\it Anisotropic turbulence:} two-dimensional slices of a set of 3D turbulence simulations~\citep{Li_2008, Perlman_2007}. 
    To create anisotropy, we have squeezed the fields along the vertical direction.
\end{itemize}
These simulations are sampled on a grid of 256$\times$256 pixels with periodic boundary conditions\footnote{When working without this condition, statistics can be computed by padding the images.} and normalized to have zero mean and unity standard deviation, respectively. Samples of each field are displayed in the first row of Fig.~\ref{fig:synthesis}. To clearly show the morphology of small-scale structures, we zoom in to a 128$\times$128 region.

\begin{figure}
    \centering
    \includegraphics[width=\columnwidth]{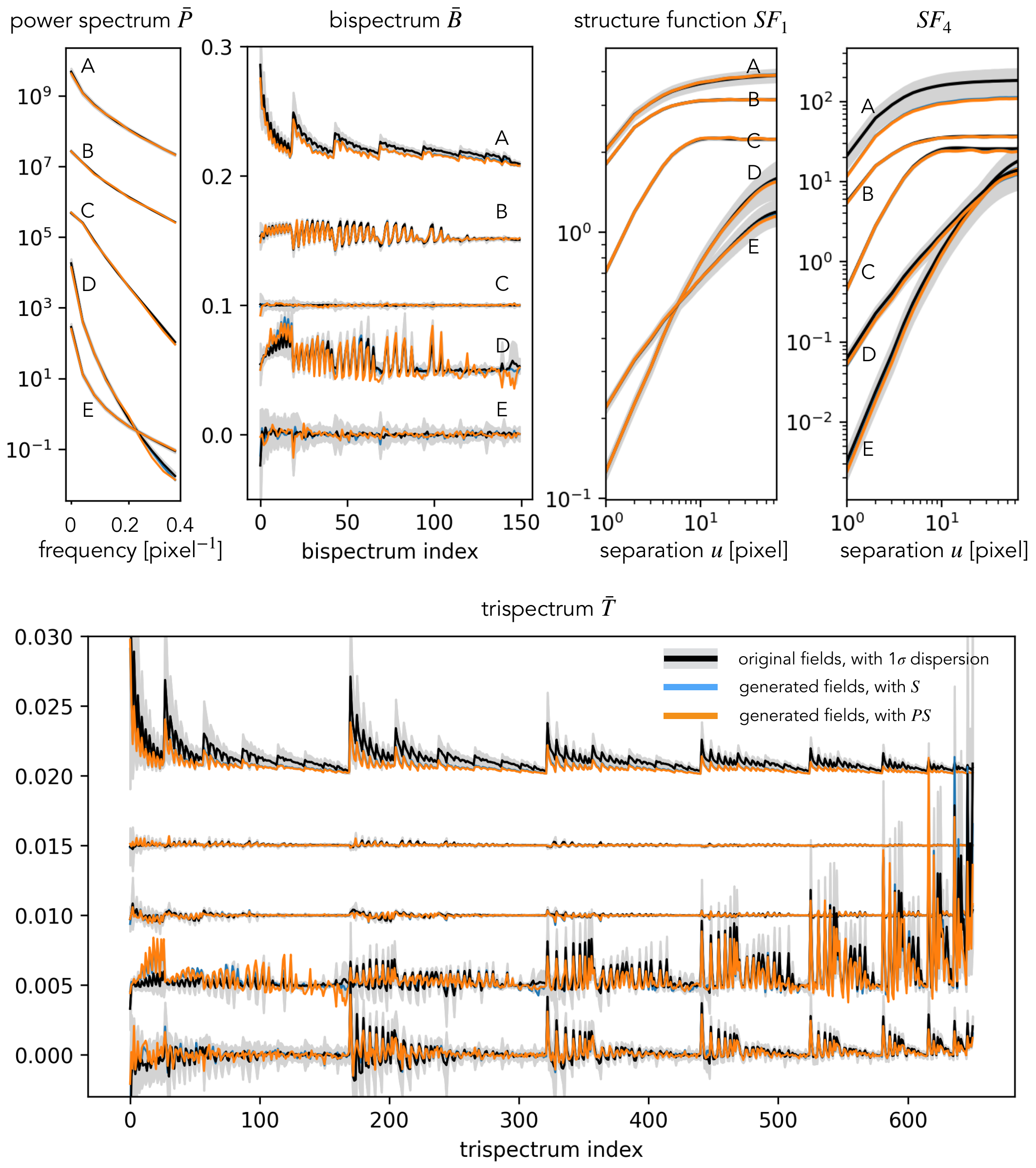}
    \caption{
    Validation of the scattering maximum entropy models for the five physical fields A--E by various test statistics. The curves for field E represent the original statistics and those for A--D are shifted upwards by an offset.
    In general, our scattering spectra models well reproduce the validation statistics of the five physical fields.
    }
    \label{fig:polyspectra}
\end{figure}

\subsection{Model description and visual validation}

We fit our maximum entropy model using wavelet polyspectra and scattering spectra, respectively, with the following constraint,
\begin{equation}
\label{eq:model}
    || {\underset {j} {\rm Ave}\,} \Phi(x_j) - {\underset {i} {\rm Ave}\,} \Phi(\sple{x}_i)||^2 \leq \epsilon
\end{equation}
where the second average is computed on an ensemble of 100 realizations $\sple{x}_i$ for each physical simulation (for field D we use only 20 realizations due to the availability of simulations), and the field generation is performed simultaneously for 10 fields $x_j$, making our microcanonical model closer to its macrocanonical limit. The microcanonical sampling algorithm is described in supplementary material~A.

Examples of field generation results are given in Fig.~\ref{fig:synthesis} The second row shows samples generated based on the high-order normalized wavelet moments $\Phi(x) = \bar M (x) = (\bar M_2(x), \bar M_3(x), \bar M_4(x))$, where $\bar M_2 = \bar S_2$, $\bar M_3(x)[\lambda, \lambda'] = \frac{M_3(x)[\lambda, \lambda']}{\sigma^2[\lambda]\sigma[\lambda']}$ and $\bar M_4(x)[\lambda, \lambda'] = \frac{M_4(x)[\lambda, \lambda']}{\sigma^2[\lambda]\sigma^2[\lambda']}$ are defined similarly to $\bar S$ in~\eqref{eq:scat-spectra-normalized}.
For the choice of wavelets, we use J=7 dyadic scales, and we set $R=8$ which samples 4 orientations within $\pi$, resulting in $\mathrm{dim}\,\bar M$ = 11\,677 coefficients for $\bar M$. The third row in Fig.~\ref{fig:synthesis} shows results from a reduced set $\Phi(x) = P\bar M(x)$, which is a 2$\sigma$ Fourier thresholded representation of $\bar M$ defined in exactly the same way as $P\bar S$ in \eqref{PS}. 
The thresholding yields $\mathrm{dim}\,P\bar M$ = 147, 286, 547, 1708, 926 for fields A--E, respectively.
A visual check shows that these models fail to recover all morphological properties in our examples especially when a thresholding reduction is applied. This issue is a manifestation of the numerical instability of high-order moments. 

In the fourth row, we present sample fields modeled with the scattering spectra $\overline{S}$ with $\mathrm{dim}\,P\bar S$ = 11\,705 for J=7 and R=8. A visual check reveals its ability to restore coherent spatial structures including clumps, filaments, curvy structures, etc. The low-order nature and numerical stability of $\overline{S}$ also significantly fasten the sampling compared to the high-order moments $\bar{M}$ (200 vs. 800 steps to converge). The last row shows sample fields modeled by a much smaller set $P\overline{S}$, which has 
$\mathrm{dim}\,P\bar S$ = 204, 364, 489, 615, 304
coefficients for fields A--E, respectively. This model is $\sim 10^2$ times smaller, while generating samples visually indistinguishable from the full set model with $\Phi(x)=\overline{S}(x)$.
In addition, the ratio between the dimensionality of the field $\mathrm{dim}\,x = L^d$ (the number of pixels) and the model $\mathrm{dim}\,\Phi$ is more than 100.
For interested readers, we also present the improvement of modeling from using power spectrum alone to the full scattering spectra in supplementary material~F. 
\subsection{Statistical Validation}
\label{sta-val}

We now quantify the consistency between the scattering spectra models and the original fields using a set of validation statistics $V(x)$ defined below, including marginal PDF, structure functions $SF_n$, power spectrum $P$, and normalized bispectrum $\bar B$ and trispectrum $\bar T$. The validation statistics are shown in Figs.~\ref{fig:synthesis} and~\ref{fig:polyspectra}, where black curves represent the expected value $\mu_\text{original}$ of these statistics, estimated from 100 realizations $\bar x_i$ of the original simulated fields (except for field D for which we have only 20 realizations).
Gray regions around the black curves represent the standard deviations $\sigma_\text{original}$ of those statistics estimated on the original fields.
Blue curves are statistics $\mu_{\bar S,\text{model}}$ estimated on fields modeled with $\overline{S}$. Similarly, $\mu_{P \bar S, \text{ model}}$ are estimated on fields modeled with the reduced set $P\overline{S}$. Both these averages are estimated from the 10 fields simultaneously sampled from the corresponding microcanonical models.

The marginal probability distribution function (PDF) is measured as the histogram of sample fields and shown in Fig.~\ref{fig:synthesis}. It averages out all spatial information and keeps only the overall asymmetry and sparsity properties of the field. The marginal information is not explicitly encoded in the scattering spectra, but for all the five physical fields we examine here, it is recovered even with the reduced model $P \bar S$, where only $\sim 10^2$ scattering spectra coefficients are used.

Given that the high dimensionality of the full set of polyspectra coefficients, as well as the computational cost of estimating them properly,
we adopt an isotropic shell binning for the power spectrum, bispectrum, and trispectrum.
Although this reduces the number of coefficients as well as their variance, working with isotropic statistics prevents the characterization of anisotropic features, for instance in fields D and E, unlike with scattering spectra. Validation results with these isotropic polyspectra are given in Fig.~\ref{fig:polyspectra}. 

The shell binning is defined as follow. We first divide the Fourier space into 10 annuli with the frequencies linearly spaced from 0 to $0.4$ cycles/pixel. Then, we average the power and poly spectra coefficients coming from the same annulus combinations. For instance, the power spectrum yields:
\begin{align}
    P[i] = {\underset {k \text{ in annuli }i} {\rm Ave}\,}\hat{x}[k]\hat{x}[-k].
\end{align}
To decorrelate the information from the power spectrum and higher orders, we normalized the binned bi- and tri-spectra by $P[i]$:
\begin{align}
    \bar B[i_1, i_2, i_3] = \frac{{\underset {k_n \text{ in annuli } i_n} {\rm Ave}\,}\hat{x}[k_1]\hat{x}[k_2]\hat{x}[k_3]}{\sqrt{P[i_1] P[i_2] P[i_3]}},
    \label{eq:B}
\end{align}
\begin{align}
\bar T[i_1, i_2, i_3, i_4] = \frac{{\underset {k_n \text{ in annuli } i_n} {\rm Ave}\,}\hat{x}[k_1]\hat{x}[k_2]\hat{x}[k_3]\hat{x}[k_4]}{\sqrt{P[i_1] P[i_2] P[i_3] P[i_4]}},
\end{align}
where the $k_n$ $d$-dimensional wave-vectors are respectively averaged in the $i_n^\text{th}$ frequency annuli, and satisfy $\sum_n k_n = 0$. 
To clearly reveal the diversity of different type of physical fields, the trispectrum $\bar T$ coefficients shown in Fig.~\ref{fig:polyspectra} are subtracted by the reference value of Gaussian white noise, evaluated numerically on 1000 independent realizations. Details about the numbers and the ordering of $\bar B$ and $\bar T$ are given in supplementary material~E.

In Fig.~\ref{fig:polyspectra} we also show the validation with structure functions, which are $n$-th order moments of the field increments as a function of the position separation $\Delta u$
\begin{align}
    SF_n[\Delta u] = \Avu \big|x[u]-x[u + \Delta u]\big|^n.
\end{align}
In our 2D case, we further average over $\Delta u$ with different orientations to obtain a structure function only depending on the magnitude of the separation $|\Delta u|$. Initially proposed by Kolmogorov for the study of turbulent flows \citep{kolmogorov1941local}, they are widely used to analyze non-Gaussian properties of multiscale processes \citep{jaffard2004wavelet}.

We quantify the discrepancy between the model and original field distributions by the outlier fraction of validation statistics outside the 2$\sigma$ range,
\begin{align}
    |\mu_\text{model} - \mu_\text{original}| / \sigma_\text{original} > 2\,.
\end{align}
For each of the five types of fields, we observe the following fractions. The binned power spectrum has fractions of $P$: 0\%, 0\%, 20\%, 0\%, 0\% for the models using all $\bar S$ statistics and 0\%, 10\%, 40\%, 10\%, 0\% for the thresholding models with $P\bar S$. The power spectrum deviation of field C is likely caused by the longer convergence steps required by smooth fields, as our generative models start from white noise with strong small-scale fluctuations. Indeed increasing the steps to 800 reduces the outlier fraction of the $P\bar S$ model to 10\%. For $\bar B$ and $\bar T$, the outlier fractions are all below 5\% except for the models of field A, where the bispectrum coefficients have 13\% of outliers. Those outliers all have the smallest scale involved, and disappear if the high-frequency cut is moved from 0.4 to 0.35 cycles/pixel. The low fractions demonstrate consistency between our maximum entropy models and ensembles of the original physical fields.

For field A, a similar deviation is also observed in high-order structure functions. For this field, it can be seen from Fig.~\ref{fig:polyspectra} that even though many coefficients are not defined as outliers, they all tend to have a lower value than the original ones. This effect may originate from the log-normal tail of the cosmic density field \citep{Coles_1991}, whose Gibbs potential includes terms in the form of $\log{x}$, in contrast to the form of $|x|$ in scattering spectra or $x^n$ in high-order statistics. However, regardless of this difficulty,
these outliers are all still within a 3$\sigma$ range, demonstrating that the scattering spectra provide a good approximation though not exact model for fields with such heavy tails.

The marginal PDF, structure functions, power spectrum and polyspectra probe different aspects of the random field $p(x)$. The polyspectra especially probe a huge variety of feature configurations. For all the validation statistics, we observe general agreement between the model and original fields. Such an agreement is a non-trivial success of the scattering spectra model, as those statistics are not generically constrained by the scattering spectra for arbitrary random fields. They indeed significantly differ from the scattering spectra in the way they combine spatial information at different frequencies and in the non-linear operation adopted. The agreement implies, as we have argued, that symmetry and regularity can be used as strong inductive bias for physical fields and the scattering spectra, with those priors build-in, can efficiently and robustly model physical fields.

\begin{figure}[t]
    \includegraphics[width=\columnwidth]{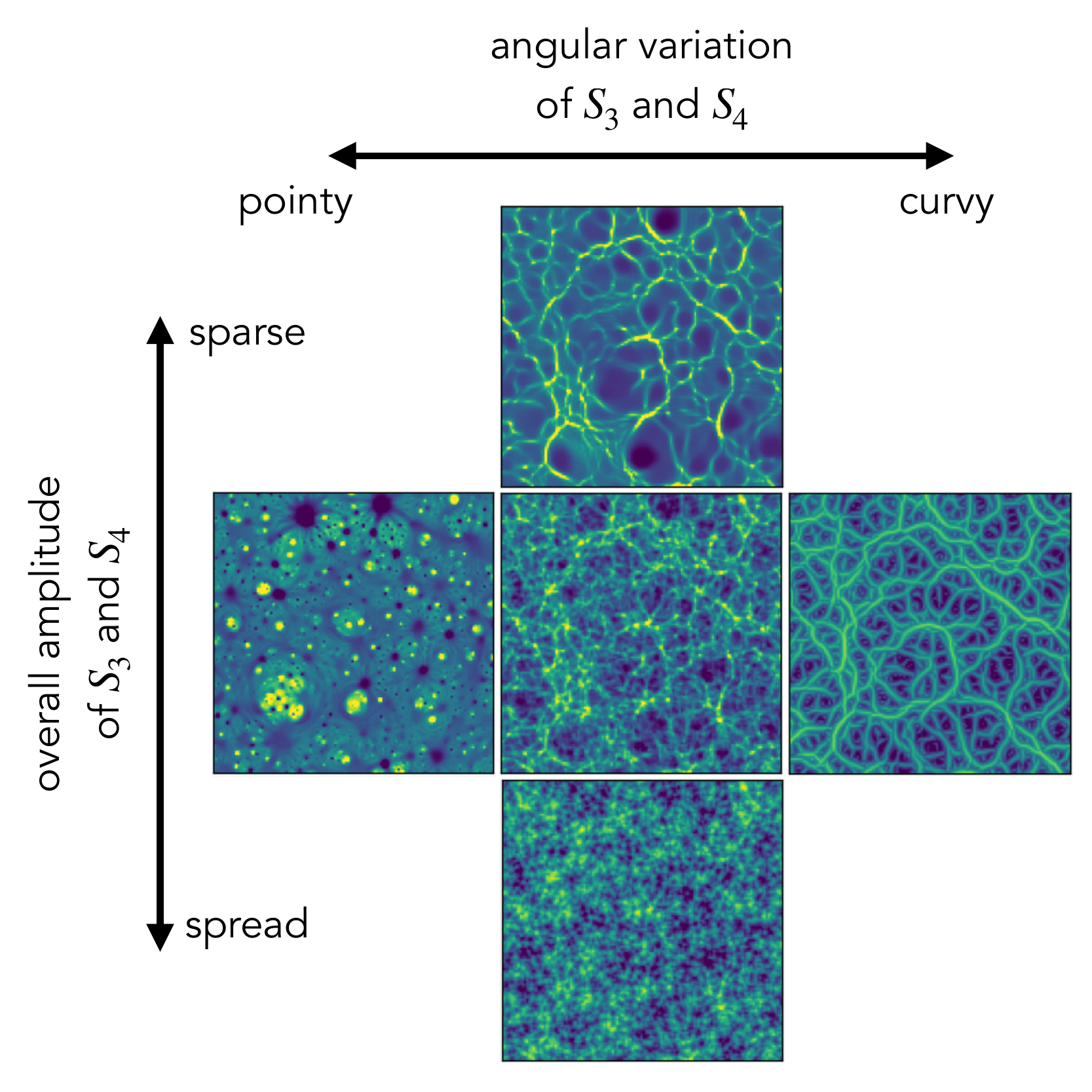}
    \caption{Visual interpretation of the scattering spectra. The central field is one realization of field B in physical simulations. The other four panels are generated fields with two simple collective modifications of the scattering spectra coefficients.}
    \label{fig:synthesis_interpretation}
\end{figure}

\subsection{Visual Interpretation of Scattering Spectra Coefficients} 

The key advantage of the scattering spectra compared to usual convolutional neural networks is their structured nature: their computation corresponds to the combination of known scales and orientations in a fixed way. Beyond the limited number of symmetries, the structured nature of the scattering spectra allows us to both quantify and interpret the morphology of structures, which is one of the original goals to design these statistics. 

The values of scattering spectra can be shown directly (see Fig.~S3) to analyze non-Gaussian properties of the field. Moreover, the meaning of its coefficients can also be visualized through our maximum entropy generative models. As one gradually changes the value of some summary statistics, the morphology of structures in the generated fields also changes. A similar exploration for a smaller set of scattering transform coefficients has been explored in \cite{cheng2021quantify}, and we show such results with the much more expressive scattering spectra coefficients in Fig~\ref{fig:synthesis_interpretation}. Such exploration using synthesis is also similar to the feature visualization efforts for convolutional neural networks \citep{olah2017feature}.

The central panel is a realization of field B from physical simulations. The other four panels are generated fields with two collective modifications of the scattering spectra: the vertical direction shows the effect of multiplying all $\bar S_3$ and $\bar S_4$ coefficients by a factor of 1/3 or 3. It indicates that the amplitude of $\bar S_3$ and $\bar S_4$ controls the overall non-Gaussian properties of the field and in particular the sparsity of its structures.
The horizontal direction corresponds to adjusting the orientation dependence. We set the coefficients with parallel wavelet configurations (i.e.,  $\bar S_3[||\lambda|, |\lambda'|, l_1=l_2]$ and $\bar S_4[|\lambda|, |\lambda'|, |\gamma|, l_1=l_2=l_3]$) as references and keep them unchanged. Then, we make the difference from other coefficients to those references to be 2 times or --2 times the original difference. Visually, it controls whether structures are more point-like or more curvy-like in the field. In this experiment, the generated field is initialized with the original field instead of white noise, in order to clearly show the correspondence between the field structure and scattering spectra coefficients.

\subsection{Application to Identifying Symmetry}

As an expressive representation whose coefficients are equivariant under standard group transformation, the scattering spectra can also be used to detect and identify the various statistical invariances commonly present in physical fields. Besides the aforementioned rotation and scaling invariance, more can also be included, such as the flipping of coordinate or field values. 

The simplest way to check asymmetry to a transformation like rotation or flip is to check if the scattering spectra $S$ are changed after applying such a transform. A more sophisticated way that can also quantify partial symmetries is to linearly decompose $\bar S$ into symmetric and asymmetric parts and then compute the fraction of asymmetric coefficients surviving the thresholding reduction. 
We further normalize this fraction by that in the full set to eliminate the dependence on image size: 
\begin{align}
    \text{asymmetry index} = \frac{\text{dim}(P{\bar S}_\text{asym})}{\text{dim}(P{\bar S})} / \frac{\text{dim}({\bar S}_\text{asym})}{\text{dim}({\bar S})}.
\end{align}
When it is zero, the random field $p(x)$ should be invariant to the transform up to the expressivity of our representation. For the five random fields analyzed in this study, we measure their asymmetry indices with respect to rotation and scaling. The corresponding anisotropy and scale dependence indices are (A) 0, 0.16 ; (B) 0, 0.53; (C) 0, 0.66; (D) 0.32, 0.45; (E) 0.28, 0.29. As expected, the cosmic lensing field (field A) is closest to isotropic and scale-free, because the scale range of the simulated field (approximately 80 Mpc in physical size) falls in the non-linear scale of cosmic structure formation and thus consists of peaks with all sizes and strengths. The cosmic web (B) and 2D turbulence (C) fields are isotropic but not scale-free, because they have particular physical scales above which the field becomes Gaussian: for cosmic web it is around 150 Mpc (25 pixel), and for turbulence it is the scale of driving force (32 pixel), both in the middle of the scale range of our simulations. The last two turbulence fields have anisotropic physical input, but the latter largely probes the `inertial regime' of turbulence, which is scale-free. 
\subsection{Limitations}

\begin{figure}
    \centering
    \includegraphics[width=\columnwidth]{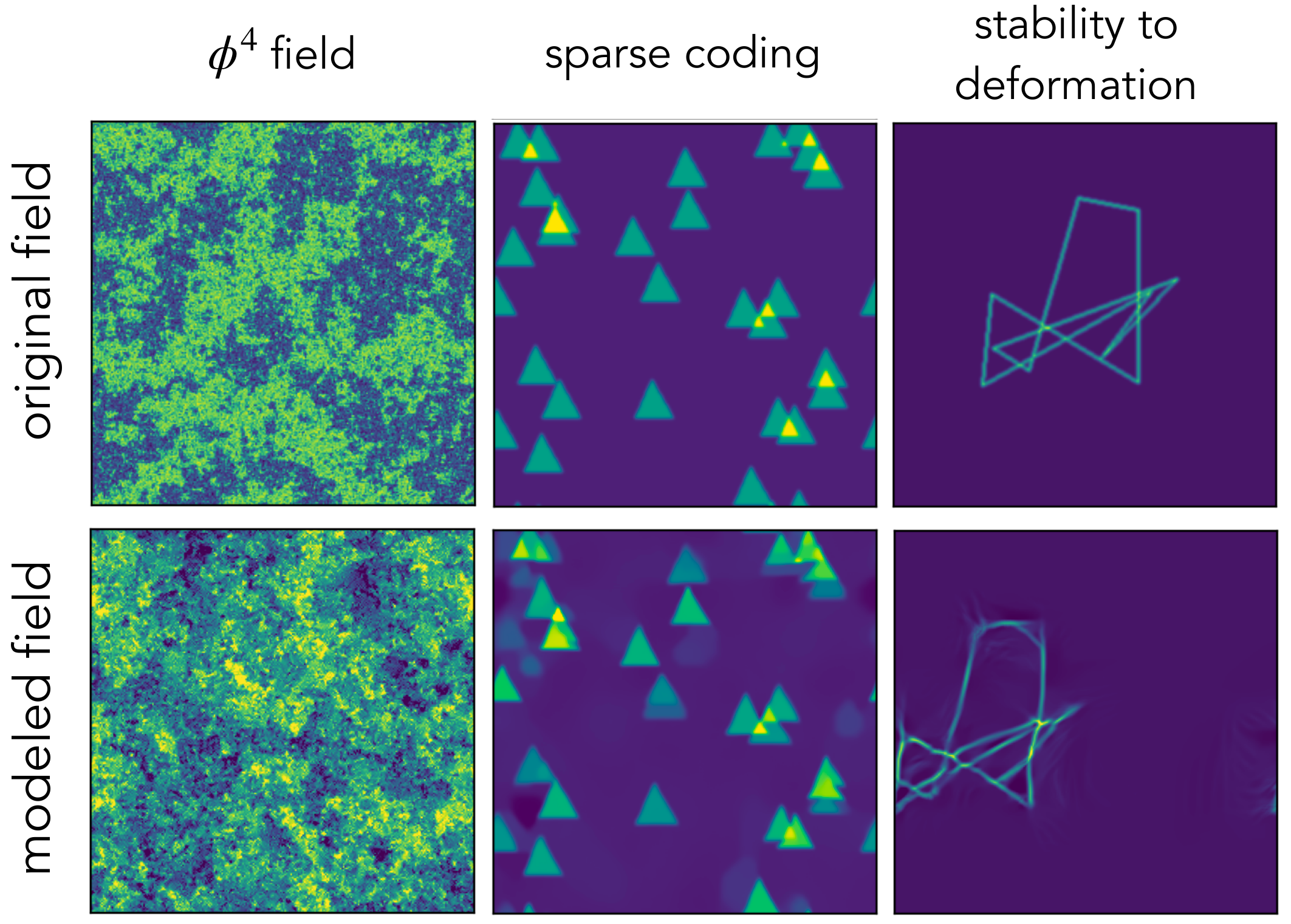}
    \caption{Example of failures and applications beyond typical physical fields. The modeled field of the central panel has been recentered for easier comparison with the original ones.}
    \label{fig:limit}
\end{figure}

While a broad range of physical fields satisfy the implicit priors of the scattering spectra, one does expect regimes for which the description will not be appropriate. The so-called $\varphi^4$ field in physics comes as a first problematic example. It is the maximum entropy field under the power spectrum and pointwise fourth-order moment $x^4$ constraints, but this characterization is unstable to specify a non-convex pdf which is a pointwise property as opposed to the delocalized Fourier moments and it is highly unstable at critical points \citep{WCRG}. The first column in Fig.~\ref{fig:limit} shows an original $\varphi^4$ field at its critical temperature and that generated from the full set of scattering spectra. In contrast to previous examples, this type of field is not successfully reproduced. 

On the other hand, when built based on one example field $x_1$ and generating only one realization $\bar{x}_1$ (i.e., in Eq.~\ref{eq:model} both $i$ and $j$ are 1), our model has a risk of over-fitting: it almost exactly copies the original field with an arbitrary translation and does not provide enough randomness. It can also be seen as a transition from generative modeling regime into a coding regime. This is related to the fact that for maximum entropy models, when the number of constraints amounts to a considerable fraction of the number of total degree of freedom, the microcanonical distribution deviates significantly from the macrocanonical distribution, and has a much lower entropy. The middle panel of Fig.~\ref{fig:limit} illustrate this effect, where the relative position of triangles of the modeled field is exactly copied from the original field. It happens only when the field is sparse, and when the full set $\bar{S}$ is used. This problem can be avoided by increasing the number of input fields or generated fields, or an early stop in the microcanonical sampling.

For physical fields with multi-scale structures, it is expected that the distribution function $p(x)$ does not change much under a slight deformation. When modeling such fields, it is important to have a representation that has the same property. Being built from wavelet decomposition and contracting operator, the scattering spectra also linearize small deformation in the field space, which plays an important role in lowering its variance~\citep{bruna2013invariant}. However, when modeling structured fields whose distribution functions are not regular under deformation, this means that the generative model will simply produce structures that are  close enough up to small deformations. This typical type of failure is shown in the third example of Fig.~\ref{fig:limit}.


\section{Conclusion}

We build maximum entropy models for non-Gaussian random fields based on the scattering spectra statistics. Our models provide a low-dimensional structured representation that captures key properties encountered in a wide range of stationary physical fields, namely:
(i) stability to deformations as a result of local conservation laws in Physics for mass, energy, momentum, charge, etc;
(ii) invariance and regularity to rotation and scaling;
(iii) scale interactions typically not described by high-order statistics;
Those are the priors included in the scattering spectra.

Our models provide a practical tool for generating mock fields based on some example physical fields. In sharp contrast to neural network models, our representation has the key advantage of being interpretable and can be estimated on a few realizations. This is crucial in Physics where generating fields in experiments or simulations is costly or when non-stationarity limits the amount of clean recorded data. 
Our proposed approach enables a new range of data/simulation analyses \citep[e.g.][]{regaldo2021new,delouis2022non}, involving extensions to the modeling of cross-regularities when multiple channels are available \citep[e.g.][]{regaldo2023generative}.

\section{Acknowledgments}
S.C. thanks Siyu Yao for her constant inspiration and encouragement.

\section{Funding}
We acknowledge funding from the French government as part of the “Investissements d’avenir” program ANR-19-P3IA-0001 (PRAIRIE 3IA Institute). S.C. acknowledges the support of the Martin A. and Helen Chooljian Member Fund, fund from Zurich Insurance Company, and the Fund for Natural Sciences at the Institute for Advanced Study. B.M. acknowledges support from the David and Lucile Packard Foundation.

\section{Author contributions statement}
S.C. developed part of the algorithm, most of the implementation and numerical experiments, introduced the correspondence to higher-order moments, and had an important writing contribution. R.M. developed some of the core mathematics, and algorithms related to self-similarity and their implementation, and wrote the corresponding sections and appendices. E.A. introduced transformations related to rotation symmetries and reduced representations and participated in numerical experiments and their expositions. B.M. provided his physical expertise to define the problem, evaluate solutions and present results in the paper. S.M. provided an important part of the mathematical framework of the representation and organized the teamwork and written paper.

\section{Competing interest statement}

Authors declare no competing interest.


\section{Preprints}
A preprint of this article is published at
\href{https://doi.org/10.48550/arXiv.2306.17210}{\url{https://doi.org/10.48550/arXiv.2306.17210}}.

\section{Data availability}

The code used in this paper can be found on the github: \href{https://github.com/SihaoCheng/scattering_transform}{\url{https://github.com/SihaoCheng/scattering_transform}}.
Data used for this work were previously published, including the Dark Matter dataset \citep{matilla2016dark, gupta2018non} from the Columbia lensing group accessible at \href{http://columbialensing.org}{\url{http://columbialensing.org}}, the Quijote Simulations \citep{quijote2020} available at \href{https://quijote-simulations.readthedocs.io}{\url{https://quijote-simulations.readthedocs.io}}, the turbulence simulations published in \cite{schneider2006coherent} and \cite{allys2019rwst}, and the forced isotropic turbulence simulation \cite{Li_2008, Perlman_2007} from the Johns Hopkins Turbulence Database accessible at \href{http://turbulence.pha.jhu.edu}{\url{http://turbulence.pha.jhu.edu}}.


\section{Supplementary Material}

\section{A. Microcanonical Sampling}
\label{app:micro-sec}

Given $n$ observed samples $\overline{x}_1,\ldots,\overline{x}_n$ of a field, with possibly $n=1$, the microcanonical ensemble given in eq. (5) can be extended as follow:
\begin{equation}
\label{micro-set1}
\Omega_\epsilon = \Big\{ x_1,\ldots,x_m \in \R^{L^d}
~:~
\|\Av{j}\Phi(x_j) -   \Avi\Phi(\overline{x}_i) \|^2 \leq \epsilon \Big\} .
\end{equation}
Microcanonical models are maximum entropy distributions over $\Omega_\epsilon$, which have a uniform distribution over this ensemble. Increasing the number of samples $n$ reduces the variance of ${\rm Ave}_i\Phi(\overline{x}_i)$ which concentrates around $\E\{\Phi(x)\}$. This reduces the information about a specific realization which is contained in ${\rm Ave}_i\Phi(\overline{x}_i)$, thus limiting over-fitting.

Sampling from the microcanonical model amounts to drawing a realization from a uniform distribution in $\Omega_\epsilon$. We approximate this sampling with a gradient descent algorithm studied in~\cite{bruna2019multiscale}. This algorithm progressively transports a white Gaussian noise distribution, which has a higher entropy than the microcanonical model, into distributions supported in $\Omega_\epsilon$. This is done with a gradient descent on $\ell(y_1,\ldots,y_m) = \| {\rm Ave}_j\Phi(y_j) - {\rm Ave}_i\Phi(\sple{x}_i) \|^2$, where the $y_j$ are initialized as independent realizations of white noises. At each iteration, the $y_i$ are updated with the L-BFGS-B algorithm, which is a quasi-Newton method that uses an estimate of the Hessian matrix. In practice, we perform $200$ gradient descent steps which yield a typical error $\epsilon\approx10^{-4}$.

It is proved in~\cite{bruna2019multiscale} that this algorithm converges to a distribution that has the same symmetries as $\Phi(x)$, similarly to the microcanonical one. However, it has been shown that this algorithm recovers a maximum entropy distribution in $\Omega_\epsilon$ only under appropriate conditions and that such gradient descent models may differ, in general, from maximum entropy ones. Nevertheless, these algorithms provide powerful sampling methods to approximate large classes of high-dimensional stationary processes, while being much faster and computationally tractable than alternative MCMC algorithms. 

\begin{figure}
	\centering
    \includegraphics[width=\columnwidth]{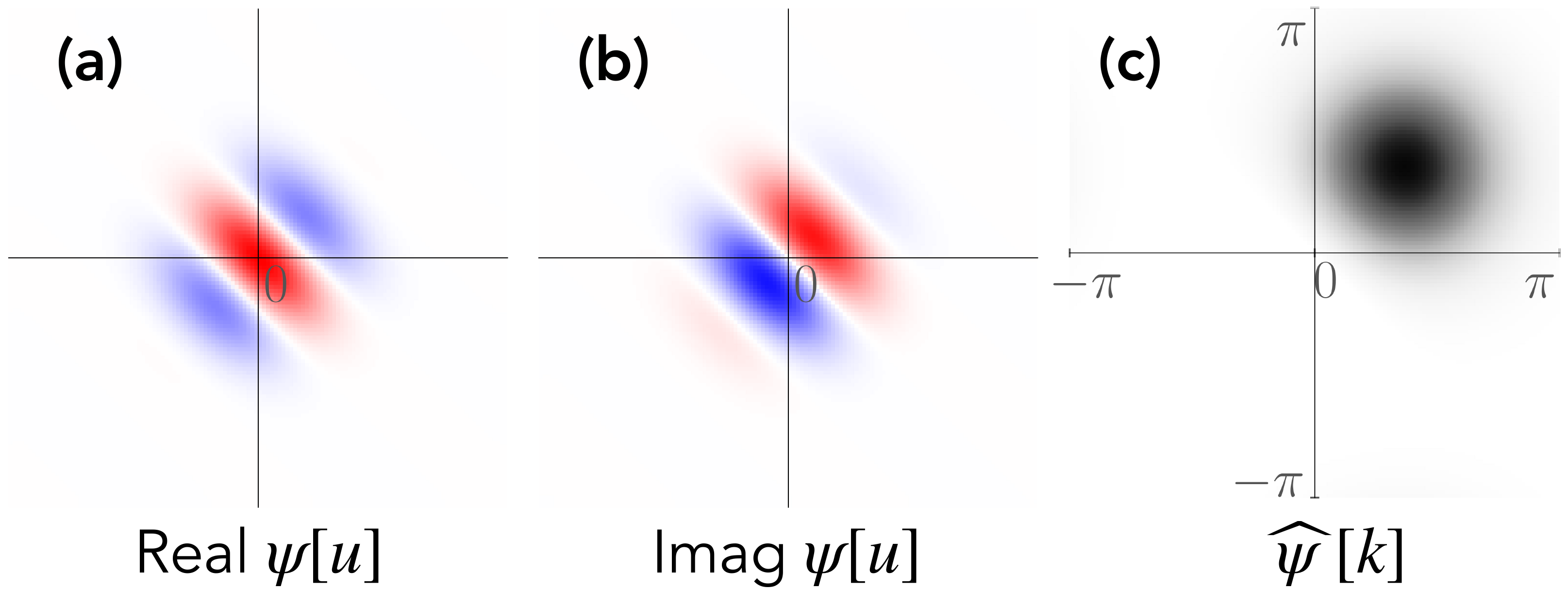}
	\caption{Real and imaginary parts of a Morlet wavelet $\psi[u]$ and its Fourier transform $\hat \psi[k]$, used in numerical calculations.}
	\label{FigWavelet}
\end{figure}

\section{B. Wavelets in $\R^d$ and Scattering Spectra}
\label{app:wavelet}

A Morlet wavelet $\psi$ defined on $\R^d$ is the product of a Gaussian envelope with a sinusoidal wave
\begin{align}
\psi[u] = g_{\sigma}[u](e^{i\langle\xi,u\rangle}-c) 
~~ \text{with} ~~ 
g_\sigma[u]
= \frac{1}{(\sigma\sqrt{2\pi})^d}e^{-\frac{\|u\|^2}{2\sigma^2}},
\end{align}
where $c$ is chosen so that $\int \psi[u]\,du = 0$. In practice, the envelope $g_{\sigma}$ is an elliptical Gaussian window to increase the angular resolution of $\psi$, but this does not virtually modify our discussion.
Such a wavelet recovers variations around scale $2^j$ in the direction of $\xi$. It is invariant to any rotation of $\R^d$ that fixes $\xi$. In practice we choose $\xi=(3\pi/4,0,\ldots,0)$ and $\sigma=0.8$.
For simplifying equations in appendices we assume $\|\xi\|=1$, without loss of generality. To recover variations at other scales and in other directions we define the wavelet filters
\[
\psi_\lambda[u] = 2^{-jd}\psi[2^{-j}r^{-1}u] ~~ \text{with} ~~ \lambda = 2^{-j} r^{-1} \xi
\]
for $(j, r) \in \R \times SO(d)$. In Fourier, $\widehat\psi_\lambda$ is a Gaussian centered in $\lambda$ substracted by a Gaussian centered in $0$ so that $\widehat\psi_\lambda[0] = 0$
\[
\widehat\psi_\lambda[k] = \widehat\psi[2^{-j}r^{-1}k] ~~ \text{with} ~~ 
\widehat\psi[k] = e^{-\frac{\sigma^2}{2}\|k-\xi\|^2} - c\,e^{-\frac{\sigma^2}{2}\|k\|^2}
\]

We shall restrict the scales $2^j$ to dyadic scales, hence taking $j$ integer, and restrict the rotations to a discrete subgroup $\Gamma$ of $SO(d)$ of order $2^d - 1$ \citep{Meyer:92c}. In dimension $d=2$ such a group can be parameterized by one angle, in dimension $d=3$ it can be parameterized by 2 angles. We write $\Lambda = \Z \times \Gamma$ the group that defines filters $\psi_\lambda$ from $\psi$.

To guarantee that the wavelet transform $W$ (defined in eq. 9) is invertible and satisfies an energy conservation,
we impose that the $\psi_\lambda$ satisfy the following Littlewood-Paley inequality for $0<\delta<1$
\begin{equation}
\label{littlewood}
\forall k \ne 0 ~~,~~
1-\delta \leq \sum_{\lambda\in\Lambda}|\widehat{\psi}_\lambda [k]|^2 \leq 1+\delta.
\end{equation}
For fields defined on a cubic $d$-dimensional lattice of length $L$, the wavelets $\psi_\lambda$ are discretized accordingly.
The wavelet transform is computed up to the largest scale $2^J$ which is smaller than length $L$ so as to achieve a reasonable estimate of low-frequency moments, even on a single realization.
The lower frequencies of $x$ in the ball $|k|\leq 2^J$ are captured by a low-pass filter $\psi_0$ which is a Gaussian centered in $k=0$ in Fourier $\widehat\psi_0[k] = c_0\,\exp(-\sigma_0^2\|k\|^2/2)$ with $\sigma_0=\sigma\,2^{J-1}$. 
The Littlewood-Paley inequality now reads:
\begin{align}
\forall k \ne 0 ~~,~~
1-\delta 
\leq |\psi_0[k]|^2 + \sum_{|\lambda|^{-1}\leq2^J}|\widehat{\psi}_\lambda [k]|^2 \leq 1+\delta.
\end{align}
By applying the Parseval formula we derive that for all $x$
\begin{align}
(1-\delta)\|x\|^2 \leq \|Wx\|^2  \leq (1+\delta)\|x\|^2
\end{align}
which insures that $W$ preserves the norm of $x$, up to a relative error of $\delta$, and is therefore invertible, with stable inverse. For the wavelet used for syntheses of physical fields in this paper, we have $\delta\approx0.8$.

Covariance of wavelet coefficients $Wx[u,\lambda]$
can be written from the power spectrum $P(x)$ of $x$
\begin{align}
\E\{Wx[u,\lambda] Wx[u',\lambda']^*\} = \frac{1}{2\pi}\int P(x)[k]\widehat\psi_\lambda[k]\widehat\psi_{\lambda'}[k] e^{i\langle u-u',k\rangle} \mathrm{d}k.
\end{align}
It implies that this correlation is zero if the supports of $\widehat\psi_\lambda$ and $\widehat\psi_{\lambda'}$ do not overlap. For the specified wavelets, as soon as $\lambda\ne\lambda'$, these supports barely overlap and $\E\{Wx[u,\lambda]\,Wx[u',\lambda']^*\}\approx0$. 
Moreover, since $x$ is stationary, the covariance $\E\{Wx[u,\lambda]\,Wx[u',\lambda']^*\}$ only depends on $u-u'$ and have a fast decay when the power spectrum $P(x)$ is regular.
Thus, even if dependencies across separate scales may exist, they are not captured by correlation.

Taking the modulus of wavelet coefficients removes complex phase oscillations and thus recenter the frequency support of $Wx[u,\lambda]$. 
Indeed, the power spectrum $P_\lambda(x)$ of $x\star\psi_\lambda$ is mostly supported in a ball $\|k-\lambda\| \leq 2^{-j}\sigma^{-1}$ which does not overlap with the Fourier support of the power spectrum $P_{\lambda'}(x)$ of $x\star\psi_{\lambda'}$. 
Taking a modulus on $x\star\psi_{\lambda'}$ eliminates the phase which oscillates at the central frequency $\lambda'$. As a consequence, the power spectrum of $|x\star\psi_{\lambda'}|$ is centered at $k=0$ and its energy is mostly concentrated in $\|k\|\leq2^{-j}\sigma^{-1}$ which now may overlap with the support of $P_W(x)[\lambda]$ as can be seen in Fig.~2. 
The power spectra of $|Wx[u,\lambda]|$ and $|Wx[u,\lambda']|$, both centered at zero, also overlap. 

We now justify taking $u=u'$ in order 3 moments given by eq.~(13).
The cross spectrum $P_{\lambda,\lambda'}(x)$ between $Wx[u,\la]$ and $|Wx[u,\la']|$ is assumed regular for the fields considered in this paper. In that case one can approximate such cross-spectrum using wavelets, which gives the moments $\E\{WWx[u,\la,\ga]\,W|Wx|[u,\la',\ga]\}$. However, the left-hand-side $WWx[u,\la,\ga]$ is negligible when $\la\ne\ga$ because Fourier support of wavelets $\psi_\la$ and $\psi_\ga$ barely overlap. The resulting coefficients
\begin{align}
\E\{WWx[u,\la,\la]\,W|Wx|[u,\la',\la]\} = \frac{1}{2\pi}\int P_{\lambda,\lambda'}(x)[k] |\widehat\psi_\lambda|^2 dk
\end{align}
average $P_{\lambda,\lambda'}(x)[k]$ in a ball $\|k\|\leq2^{-j}\sigma^{-1}$ through $|\widehat\psi_\lambda|^2$. However, $P_{\lambda,\lambda'}(x)$ is already concentrated in this ball. 
We thus remove $|\widehat\psi_\lambda|^2$ which yields  $\E\{Wx[u,\la]\,|Wx|[u,\la']\}$.

The following proposition shows that Scattering Spectra reveal non-Gaussianity in a field $x$.
\begin{proposition}
Let $x$ be a stationary process.
\begin{enumerate}
\item If $x$ is Gaussian then for any separate scales $\lambda,\lambda'$, meaning that $\widehat\psi_\lambda\widehat\psi_{\lambda'}=0$
\begin{align}
\E\{\bar S_1(x)\}=\frac{\pi}{4},
\end{align}
\begin{align}
\E\{\bar S_3(x)[\lambda,\lambda']\}=0 ~~\text{and}~~\E\{\bar S_4(x)[\lambda,\lambda',\gamma]\}=0.
\end{align}
\item If $x$ is symmetric i.e. $p(-x) = p(x)$ then 
\begin{align}
\E\{\bar S_3(x)\}=0.
\end{align}
\item If $x$ is invariant by rotation of angle $\pi$ i.e. $p(x[-u])=p(x[u])$ then
\begin{align}
\mathrm{Im}\,\E\{\bar S_3(x)\} = 0 ~~\text{and}~~ \mathrm{Im}\,\E\{\bar S_4(x)\} = 0.
\end{align}
\end{enumerate}
\end{proposition}
\begin{proof} If $x$ is Gaussian then $Wx[u,\lambda]$ is also Gaussian and the ratio between its first and second order moment is $\pi/4$. If $\widehat\psi_\lambda\widehat\psi_{\lambda'}=0$ then $Wx[u,\lambda]$ and $Wx[u,\lambda']$ are decorrelated, since $(Wx[u,\lambda],Wx[u,\lambda'])$ is Gaussian, this implies that $Wx[u,\lambda]$ and $Wx[u,\lambda']$ are independent. Thus, $Wx[u,\lambda]$ and $|Wx[u,\lambda']|$ are independent, so are  $W|Wx|[u,\lambda,\gamma]$ and $W|Wx|[u,\lambda',\gamma]$ which proves 1.
Point 2. is proved by observing that $S_3(-x)=S_3(x)$ and point 3. by observing that $S_3(x[-u]) = S_3(x)^*$ and $S_4(x[-u]) = S_4(x)^*$.
\end{proof}
For the physical fields studied in this paper, such coefficients are non-zero, thus revealing their non-Gaussianity~Fig.\ref{fig:scattering-spectra-visu}.

\begin{figure}
\centering
\includegraphics[width=\columnwidth]{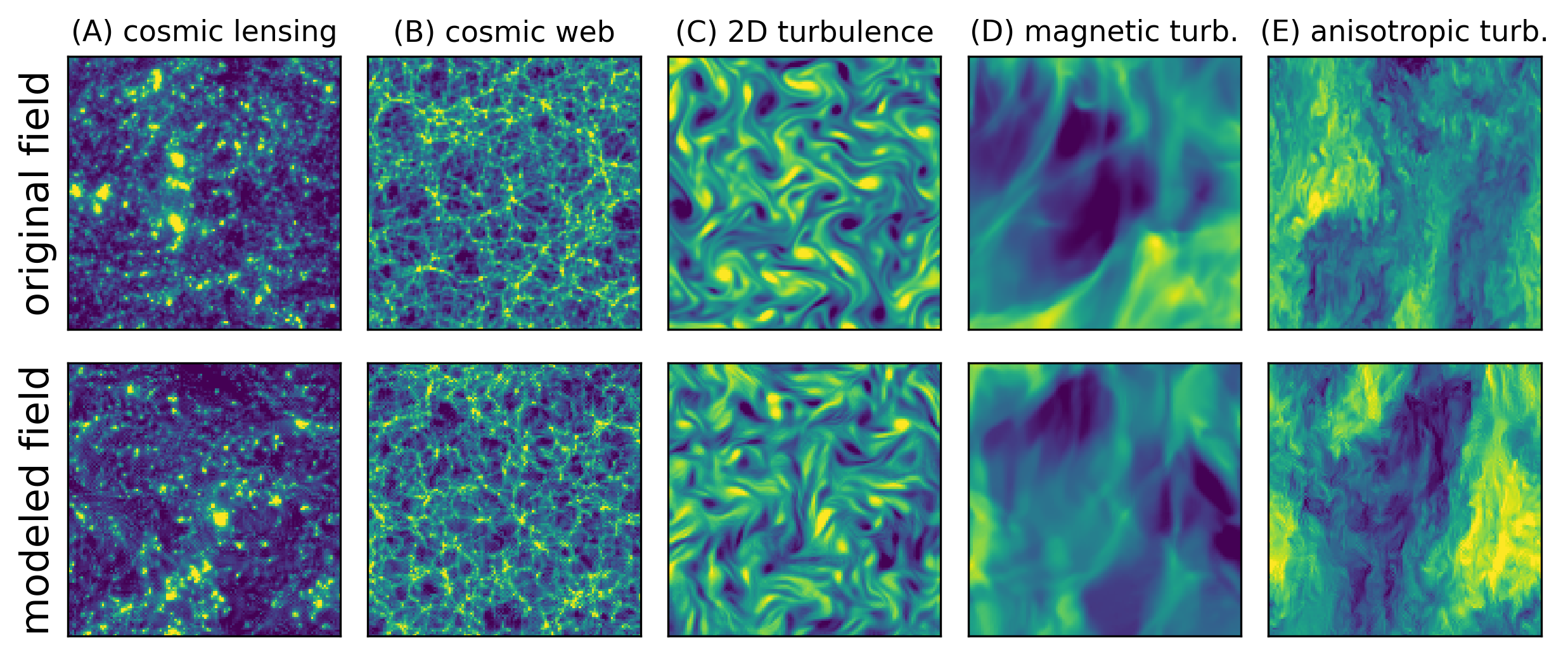}
\caption{Visual assessment of our model based on $\bar S$ with 11\,641 coefficients estimated on a single realization (top). Generated fields (bottom) show very good visual quality.}
\label{fig:synthesis-single-sample}
\end{figure}

\section{C. Equivariance and Invariance to rotations and scaling}
\label{app:rotation_scaling}

The scattering spectra are computed from wavelet transforms, which are equivariant to rotations and scalings. We show that scattering spectra inherit these equivariance properties. If $p(x)$ is isotropic or self-similar, then one can build isotropic or self-similar maximum entropy models by averaging renormalized scattering spectra over rotations or scales, which reduces both the variance and dimensionality of $\bar S$.

To avoid discretization and boundary issues for rotations and scaling, we consider fields $x[u]$ defined over continuous variables $u \in \R^d$, and establish the mathematical results in this framework. For this purpose, the sum in the wavelet transform defined in eq.~(9) is replaced by an integral over $\R^d$. Wavelets are dilated by $2^j$ for $j \in \Z$ and rotated by $r$ in a rotation group $G$ of cardinal $R$. In dimension $d = 2$, these rotations have an angle $2 \pi \ell / R$.
\begin{proposition}
For $r\in G$ with $x_r [u] = x[r^{-1} u]$ one has
\begin{equation}
\label{eq:equivar-rotation}
S(x_r)[\la,\la',\ga] = S(x)[r\la,r\la',r\ga].
\end{equation}
For $j\in\Z$ with $x_j [u] = x[2^{-j} u]$ one has
\begin{equation}
\label{eq:equivar-scaling}
S(x_j)[\lambda,\lambda',\gamma] = S(x)[2^j\lambda,2^j\lambda',2^j \gamma].
\end{equation}
\end{proposition}
\noindent\textbf{Proof.} 
It follows from the equivariance of wavelet coefficients, $Wx_r[u,\lambda] = Wx[r^{-1}u,r\lambda]$ and $Wx_j[u,\lambda] = Wx[2^{-j}u,2^j\lambda]$.

Isotropic fields $x$ have a distribution that is invariant to rotation $x_r \overset{d}{=}x$ for all $r\in G$. Self-similar fields $x$ have a distribution that is invariant to scaling, up to random multiplicative factors $x_j\overset{d}{=}A_jx$ for all $j\geq0$ \cite{mandelbrot1997multifractal}. For such fields, we show that the expected scattering spectra exhibit invariance to rotation or scaling of their indices, and thus have a lower-dimensional structure. For that purpose we used normalized scattering spectra coefficient $\bar{S}(x)$ defined eq.~(22), where the normalization is done by $\sigma^2 [\la] = \E \{|Wx[u,\la]|^2\}$. 

\begin{proposition}
If $x$ is isotropic then for any $r \in G$
\begin{equation}
\label{eq:invar-rotation}
\E \{ \bar S (x)[r\lambda,r\lambda', r\gamma] \} =
\E \{ \bar S (x)[\lambda,\lambda',\gamma] \} .
\end{equation}
If $x$ is self-similar at scales $2^j \leq 2^J$ then
\begin{equation}
\label{eq:invar-scaling-S12}
\E\{S_1(x)[\la]\}=c_1|\lambda|^{-\zeta_1}~~, ~~ \E\{S_2(x)[\la]\}=c_2|\lambda|^{-\zeta_2}
\end{equation}
\begin{align}
\label{eq:invar-scaling-S3}
\E\{\bar S_3(x)[2^j\la,2^j\la']\} & = \E\{\bar S_3(x)[\la,\la']\}
\\
\label{eq:invar-scaling-S4}
\E\{\bar S_4(x)[2^j\la,2^j\la',2^j\ga]\} & = \E\{\bar S_4(x)[\la,\la',\ga]\}
\end{align}
\end{proposition}
\noindent\textbf{Proof.} Let us assume $x$ is isotropic $x_r \overset{d}{=}x$. It implies that $\E\{S(x_r)\}=\E\{S(x)\}$.
Thanks to the equivariance property of \eqref{eq:equivar-rotation} one gets the invariance property on $S$: $\E\{S(x)[r\la,r\la',r\ga]\}=\E\{S(x)[\la,\la',\ga]\}$. 
We obtain \eqref{eq:invar-rotation} by dividing this equation by $\E\{S_2(x)[r\la]\}=\E\{S_2(x)[\la]\}$.
\\
Let us assume $x$ is self-similar, $x_j \overset{d}{=}A_j x$. 
In that case one has $A_{j+j'}\overset{d}{=}A_jA_{j'}$, taking order $1$ and order $2$ moments, this implies $\E\{A_j\}=2^{-j\zeta_1}$ and $\E\{A_j^2\}=2^{-j\zeta_2}$ for certain power-law exponents $\zeta_1,\zeta_2$.
Now from self-similarity and equivariance property given by \eqref{eq:equivar-scaling} one has $\E\{S_1(x)[2^j\la]\}=\E\{A_j\}\E\{S_1(x)[\la]\}=2^{-j\zeta_1}\E\{S_1(x)[\la]\}$. Taking $2^{-j}=|\lambda|$ one obtains $\E\{S_1(x)[\la]\}=c_1|\la|^{-\zeta_1}$ with $c_1=\E\{S_1(x)[|\la|^{-1}\la]\}$ independent on $|\la|$. With the same reasoning on $S_2$ we obtain \eqref{eq:invar-scaling-S12}.
From self-similarity and equivariance property given \eqref{eq:equivar-scaling}, we get similarly:
$\E\{S_3(x)[2^j\la,2^j\la']\}=2^{-j\zeta_2}\E\{S_3(x)[\la,\la']\}$. Dividing by $\E\{S_2(x)[\la]\}=c_2|\la|^{-\zeta_2}$ yields \eqref{eq:invar-scaling-S3}. We obtain \eqref{eq:invar-scaling-S4} similarly, which proves the proposition.

The wavelet coefficient renormalization is necessary to ensure that the scattering spectra are invariant to scaling. As explained in \cite{WCRG}, it is directly related to Wilson renormalization, which yields macrocanonical parameters (physical couplings) that remain constant across scales (fixed point) at phase transitions, where the field becomes self-similar. 

If $x$ is isotropic, then \eqref{eq:invar-rotation} implies
that 
\begin{equation}
\label{rotatAve}
\Avr \bar S (x)[r\lambda,r\lambda', r\gamma]
\end{equation} 
is an unbiased estimator of $\E\{\bar S (x)[\lambda,\lambda', \gamma]\}$ with lower variance than $\bar S (x)[\lambda,\lambda', \gamma]$. 
Choosing $\Phi(x) = \Avr \bar S (x)[r\lambda,r\lambda', r\gamma]$ also reduces the dimension of our model by a factor $R$. Since this representation is invariant to rotations of $x$ in $G$, the macrocanonical and microcanonical models defined from it are also invariant to these rotations.

Similarly, if $x$ is self similar on a range of scales $2^j\leq2^J$, then \eqref{eq:invar-scaling-S3} and \eqref{eq:invar-scaling-S4} implies that 
\begin{equation}
\label{scaleAve}
\Avj \bar S_3 (x)[2^j\la,2^j\la']
~~,~~ 
\Avj \bar S_4 (x)[2^j\la,2^j\la', 2^j\ga]
\end{equation}
where the average is taken on all scales $j$ such that $ (2^j|\la|)^{-1}\leq 2^J,(2^j|\la'|)^{-1}\leq 2^J,(2^j|\ga|)^{-1}\leq 2^J$,
are unbiased estimators of $\E\{\bar S_3 (x)[\la,\la']\}$ and $\E\{\bar S_4 (x)[\la,\la', \gamma]\}$ with lower variance than $\bar S_3(x)$ and $\bar S_4(x)$. Choosing $\Phi(x) = (\bar S_1(x),\bar S_2(x),\Avj\bar S_3(x),\Avj\bar S_4(x))$ reduces the dimension of our model by at most a factor $\log L$. The resulting maximum entropy model is not necessarily self-similar due to the presence of scale-dependent moments $\E\{\bar S_1(x)\}$ and $\E\{\bar S_2(x)\}$. However, if $\bar S_1(x)[\la]$ and $\bar S_2(x)[\la]$ have a power-law decay along $\la$ our model becomes self-similar. 

\begin{figure}
\centering
\includegraphics[width=0.97\linewidth]{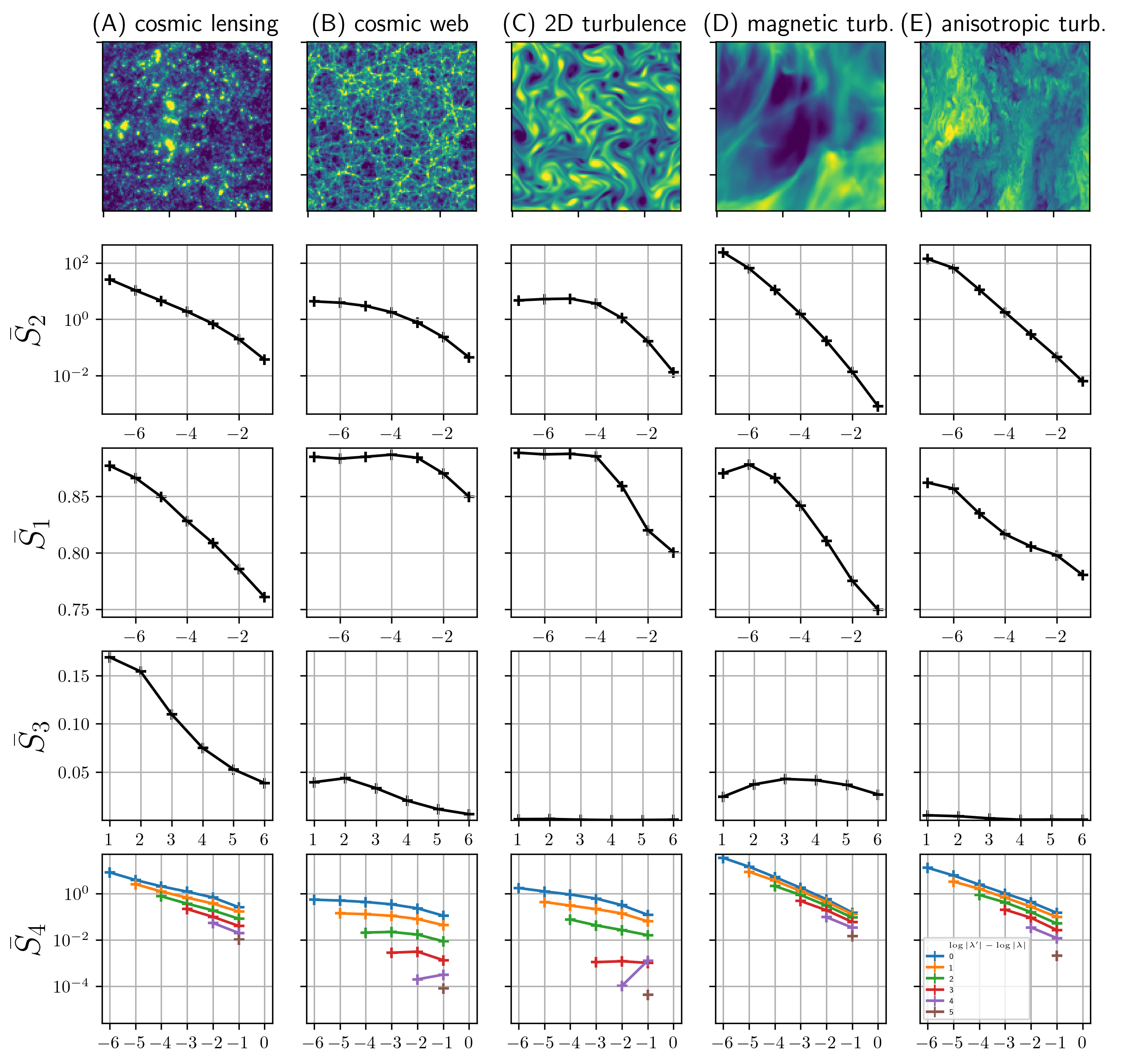}
\caption{
Visualization of Scattering Spectra $\bar S$ for different physical fields. Power-spectrum $S_2$ and sparsity factors $\bar S_1$ are averaged along all angles (amount to taking the 0-th angle Fourier harmonic).
We only show the 0-th angle Fourier harmonic and 0-th scale Fourier harmonic for order 3 and order 4-moment estimators $\bar S_3$ and $\bar S_4$. Thus, the quantities that are shown are invariant to the rotation of the field, and the last two rows ($\bar S_3, \bar S_4$ are furthermore invariant to scaling). Non-zero coefficients $\bar S_3$ show that the cosmic lensing and cosmic web fields are not invariant to sign flip. 
This is due to the presence of high positive peaks on the former and filaments on the latter. The large amplitude of envelope coefficients $\bar S_4$ on the last 2 fields indicate long-range spatial dependencies as evidenced by the presence of structures at the level of the map.
}
\label{fig:scattering-spectra-visu}
\end{figure}
%

\section{D. Dimension reduction with Fourier thresholding}
\label{app:fourier_thresholding}

We give in this appendix the details of the dimensional reduction of $\bar S$ into $P \bar S$, which is done by Fourier projectors of $\bar S(x)$ along rotations and scales, estimated by thresholding. This dimensional reduction based on regular variations of the dependence of $\bar S$ on different scales, allows for a representation of lower variance, bringing the microcanonical and macrocanonical models closer together.

We concentrate on the two-dimensional case $d = 2$ corresponding to numerical applications. The rotation group is then Abelian and defined by a single angle parameter, which simplifies the Fourier transform calculation. However, the same approach applies to 
non-commutative groups $G$ of rotations in $\R^d$ for $d > 2$, with their Fourier transform. Each wavelet frequency is defined in eq.~(8) by $\la = 2^{-j} r_\ell \xi$, where $r_\ell$ is a rotation of angle $2 \pi \ell / R$. To guaranty that the scattering spectra frequencies
satisfy $|\la| \leq |\la'| < |\ga|$, we
write 
\begin{align}
    \la = 2^{-j_1} r_{\ell_1} \xi~,~\la' = 2^{-j_1 - a} r_{\ell_2} \xi~,~\ga = 2^{-j_1-b} r_{\ell_3} \xi
\end{align}
with $0 \leq a < b \leq J - j_1$ and $J <  \log L$. 
It leads to a scale and angle reparametrization of the scattering spectra: 
\begin{align}
    \bar S(x) [\la,\la',\ga] = \bar S(x) [j_1,a,b,\ell_1,\ell_2,\ell_3].
\end{align}

If $\bar S(x)$ has regular variations as a function of rotations then its three-dimensional Fourier transform along the $(\ell_1,\ell_2,\ell_3)$ has coefficients of negligible amplitude at high frequencies, which can thus be eliminated. One can also take advantage of regularities along scales. Since $1 \leq j_1 \leq J$ varies on an interval without periodicity, the Fourier transform is replaced by a cosine transform along $j_1$ for $a$ and $b$ fixed. We could also perform a cosine transform along the scale shift $a$ and $b$, but this is not done in numerical applications because their range of variations is small and $j$-dependent. The Fourier transforms along $j_1$ is however sufficient to identify scale-invariance, since one then expects $\bar S$ to only depend on $a$ and $b$, see appendix~\ref{app:rotation_scaling}. We write $F\, \bar S(x)$ the Fourier transform of $\bar S(x)$ along $(\ell_1,\ell_2,\ell_3)$ and its cosine transform along $j_1$. 

Since $F$ is unitary, it preserves the estimator variance:
\begin{equation}
    \label{varerror}
    \sigma_{\bar S}^2 =
 \E\{ \|\Avi F \bar S(x_i) - \E\{F \bar S(x)\}\|^2\} .
\end{equation}
Ideally, the estimation error of $\E\{F \bar S(x)\}$ is reduced by eliminating its coefficients whose squared amplitude is smaller than the variance of the empirical estimation error. It amounts to suppressing all coefficients having a variance that is larger than the bias resulting from their elimination. However, we can not implement this optimal "oracle" decision because we do not know $\E\{F \bar S(x)\}$. In this paper, we instead apply an approximate thresholding algorithm, which eliminates small amplitude coefficients of $\bar S(x)$ below a threshold proportional to their standard deviation, as discussed in the main text. This thresholding algorithm is adaptive and the selected coefficients vary from one process to another. For each process studied, an ensemble of between 20 to 100 samples $\{ x_i\}$ were used to empirically estimate  the average and variance of $F\bar{S}$, called $\mu(F\bar{S})$ and $\sigma(F\bar{S})$. The coefficients which have been kept are those that individually verify $\mu(F\bar{S}) > 2\sigma(F\bar{S})$.

A projected scattering spectra
\begin{align}
    \Phi(x) = P\, \bar S(x)
\end{align}
is computed with a linear Fourier projection $P$ which eliminates all coefficients of $F \bar S(x)$ corresponding to coefficients of ${\rm Ave}_i F \bar S(x_i)$ below their threshold.
The efficiency of this projected scattering is the variance reduction ratio $\sigma_{P \bar S}^2 / \sigma_{\bar S}^2$ with
\begin{equation}
    \label{signonsfa}
 \sigma_{P \bar S}^{2} =   {\E\{ \|\Avi P  \bar S(x_i) - \E\{P \bar S(x)\}\|^2\}} . 
 \end{equation}
If $p(x)$ is isotropic or self-similar then we expect that $P$ is a low-frequency projector along global rotations (which act similarly on all $l_i$ coordinates) or scalings (which act on $j$), which corresponds to the averages described in \eqref{rotatAve} and \eqref{scaleAve}. The Fourier projection $P$ is however much more general and can adapt to unknown regularities of $p(x)$ along rotations and scales. 

\section{E. Number of coefficients for shell binned polyspectra}
\label{app:poly}

For a 2D field, there are originally $O(L^4)$ bispectrum coefficients in total, as there are two independent frequencies in the bispectrum and each has two dimensions. If we take $N_\text{bin}$ linear frequency bins along each side of $L$ lattice points, the coefficients to be estimated is reduced to $O(N_\text{bin}^4)$. A rotation and parity average will further reduce and better estimate the bispectrum coefficients, which eliminates one dimension and leads to
$\sim \frac{1}{A^3_3} \cdot \frac{1}{2}N_\text{bin} \cdot \frac{3}{4}N_\text{bin}^2 \cdot \frac{1}{2} = \frac{1}{32}N_\text{bin}^3$
binned coefficients, where $1/A^3_3=\frac{1}{6}$ is the repeated counting of the three-frequency combinations in bispectrum, $\frac{1}{2}N_\text{bin}$ is the number of choice of $k_1$, given rotation invariance, $\frac{3}{4}$ is the number of choice of $k_2$ given the requirement that each $k$ is within the $L\times L$ lattice in Fourier space and $k_1+k_2+k_3 = 0$, and the factor $\frac{1}{2}$ comes from parity average.

For 2D fields, the shell-binned bispectrum is essentially a fast way to compute the rotation and parity average of the bispectrum. It does not mix very different configurations,  because a given set of $|k_1|,|k_2|,|k_3|$ combined with the condition  $k_1+k_2+k_3=0$ uniquely set the configuration up to free rotations. The number of coefficients is of the order $\sim \frac{1}{8}N_\text{bin}^3$ (the scaling power is 3 rather than $2d=4$ because the orientation average eliminates one degree of freedom). For our choice of $N_\text{bin}=10$, there are 151 shell-binned bispectrum coefficients. Similarly, the shell-binned trispectrum $\bar T$ has  651 $\sim \frac{1}{16}N_\text{bin}^4$ coefficients.
Note that the shell-binning for trispectrum is more aggressive, because in 2D the same set of $|k_1|,|k_2|,|k_3|,|k_4|$ may come from different combinations $k_1,k_2,k_3,k_4$ even if the condition $k_1+k_2+k_3+k_4=0$ is applied.

The ordering of $\bar B$ and $\bar T$ shown in Fig.~4 is determined in a nested way. The frequency annuli are labeled by $i$ from small to large $|k|$. To remove redundant coefficients, we require $i_1 \leq i_2 \leq i_3 (\leq i_4)$ and order them first by $i_1$ in increasing order; when two binning configurations have the same $i_1$, they are then ordered by $i_2$ and so on.

\section{F. Improvement from power spectrum to scattering spectra}
\label{app:PB}

In Figure~3 we have shown the modeled fields with up to 4th order moments and scattering spectra. Here we demonstrate step by step the improvement from the traditionally used power spectrum to the powerful $\overline{S}$ in Figure~\ref{fig:PB}. The power spectrum (2nd row) is clearly not able to reproduce any non-Gaussian structures, it is able to capture the strong anisotropy in field E as we measure it with 4 different orientations. The next row shows results with isotropic bispectrum defined as eq.~(26). We modified the $|k|$ bins to be logarithmically spaced because a linear binning causes a convergence issue of the algorithm. The sparse peaks in field A start to emerge, but characteristic structures in other fields such as filaments, swirls, and stream lines are still missing.  The fourth row demonstrates the improvement from replacing the bispectrum by the selected 3rd order wavelet moments $M_3$, with which richer structures emerges in field B and C. There is a convergence issue for field D which demonstrates the numerical instability of higher-order moments. The last two rows then show the improvements from replacing higher-order moments to scattering spectra, up to 3rd and 4th order.

\begin{figure}
\centering
\includegraphics[width=\linewidth]{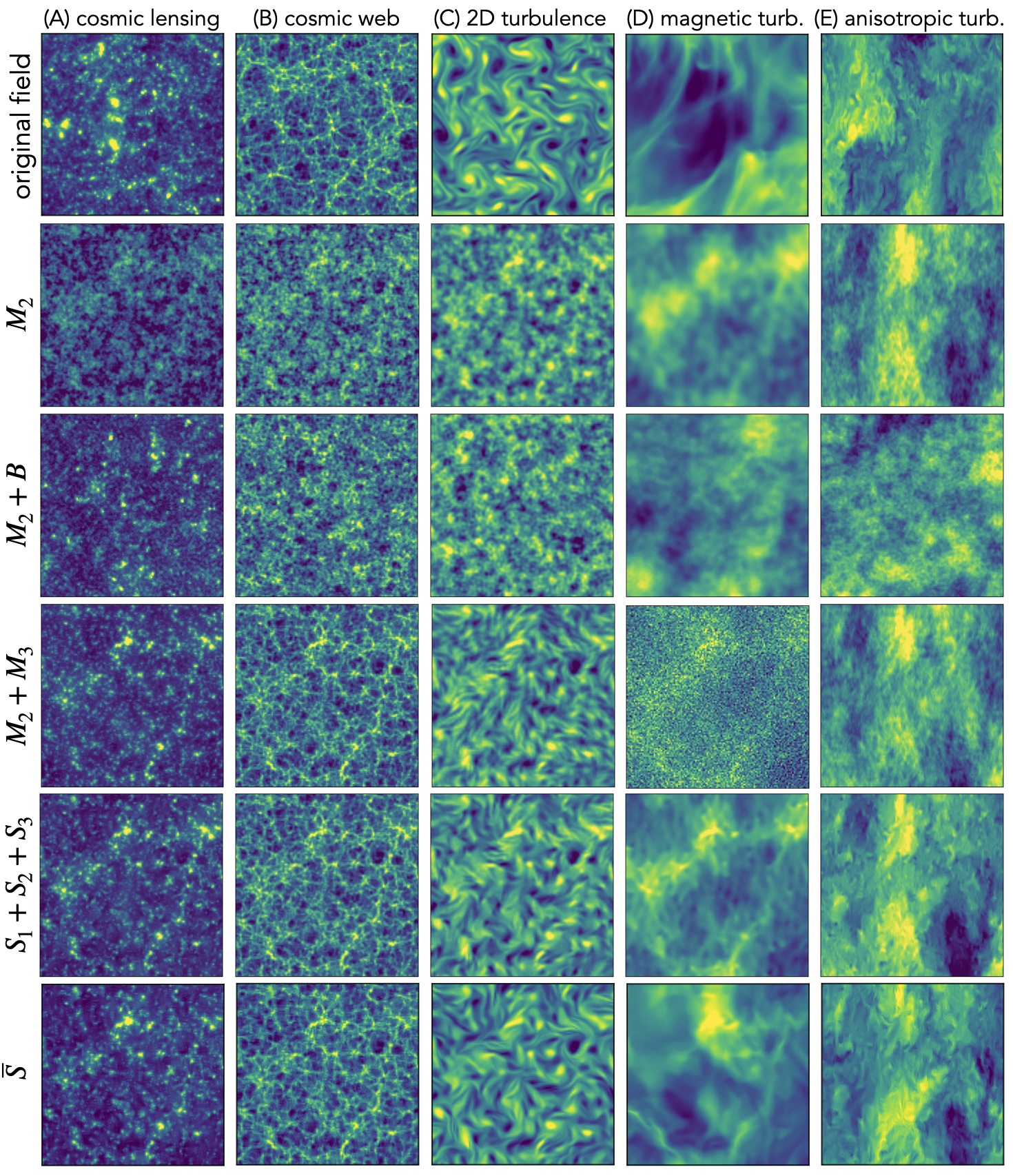}
\caption{The improvement of modeled fields from using only the power spectrum and bispectrum to the scattering spectra.
}
\label{fig:PB}
\end{figure}

\section{G. Logarithm in field B}
\label{app:log}

In our main text we have built the scattering spectra model on the logarithm of cosmic density field (field B) to characterize the cosmic web structure. An interesting question arises as whether the model performs equally well without the logarithm, where the fluctuations becomes dominated by a few high peaks and the cosmic web structure is overshone by the peaks in terms of amplitude. We have performed experiments to explore this regime. We find the scattering spectra model can reproduce many peaks and the bright filament structures, but much worse for the highest peaks and faint filaments than on the field with logarithm transform. In addition, although it can reproduce a reasonably well fit skewed PDF with a fast drop on the left and long tail on the right, it cannot completely forbid negative values; the very few highest peaks are also not bright enough. 
Given that the morphology of a web requires good characterization of the voids which have density close to zero, it is no wonder that our model does not well reproduce the faint filaments. For comparison, although the cosmic lensing field (field A) is a projected cosmic density field and also has a log-normal PDF, it is much less skewed due to the physical projection along the line of sight. As a result, although we do also observe some deviation of higher-order moments between the model and original fields caused by the log-normal PDF, it is not as severe a problem as the field B before logarithm. Also, the much flatter decrease on the low end of PDF in field A reflects the fact that cosmic web structures are averaged out due to projection.

We find that the failure of modeling field B without logarithm can be fixed by either imposing the logarithm as we have done in the main text, which guarantees the cutoff at zero and makes whole PDF more Gaussian \citep[e.g.,][]{Neyrinck_2009}, or by adding coefficients to our model which explicitly constrain the PDF profile in a differentiable manner, using e.g., smoothed histogram \cite{WCRG} or the mean values of different quantiles. However, the convergence time with the latter solution is increased by a factor of ten and requires balance between the loss function from the PDF and scattering constraints. A thorough study on incorporating the PDF constraints to the scattering spectra model is left for future work.

\bibliographystyle{mnras}
\bibliography{SC}

\providecommand{\noopsort}[1]{}\providecommand{\singleletter}[1]{#1}%
\begin{thebibliography}{}
\makeatletter
\relax
\def\mn@urlcharsother{\let\do\@makeother \do\$\do\&\do\#\do\^\do\_\do\%\do\~}
\def\mn@doi{\begingroup\mn@urlcharsother \@ifnextchar [ {\mn@doi@} {\mn@doi@[]}}
\def\mn@doi@[#1]#2{\def\@tempa{#1}\ifx\@tempa\@empty \href {http://dx.doi.org/#2} {doi:#2}\else \href {http://dx.doi.org/#2} {#1}\fi \endgroup}
\def\mn@eprint#1#2{\mn@eprint@#1:#2::\@nil}
\def\mn@eprint@arXiv#1{\href {http://arxiv.org/abs/#1} {{\tt arXiv:#1}}}
\def\mn@eprint@dblp#1{\href {http://dblp.uni-trier.de/rec/bibtex/#1.xml} {dblp:#1}}
\def\mn@eprint@#1:#2:#3:#4\@nil{\def\@tempa {#1}\def\@tempb {#2}\def\@tempc {#3}\ifx \@tempc \@empty \let \@tempc \@tempb \let \@tempb \@tempa \fi \ifx \@tempb \@empty \def\@tempb {arXiv}\fi \@ifundefined {mn@eprint@\@tempb}{\@tempb:\@tempc}{\expandafter \expandafter \csname mn@eprint@\@tempb\endcsname \expandafter{\@tempc}}}

\bibitem[\protect\citeauthoryear{Allys, Levrier, Zhang, Colling, Regaldo-Saint~Blancard, Boulanger, Hennebelle  \& Mallat}{Allys et~al.}{2019}]{allys2019rwst}
Allys E.,  Levrier F.,  Zhang S.,  Colling C.,  Regaldo-Saint~Blancard B.,  Boulanger F.,  Hennebelle P.,   Mallat S.,  2019, Astronomy \& Astrophysics, 629, A115

\bibitem[\protect\citeauthoryear{{Allys}, {Marchand}, {Cardoso}, {Villaescusa-Navarro}, {Ho}  \& {Mallat}}{{Allys} et~al.}{2020}]{Allys_2020}
{Allys} E.,  {Marchand} T.,  {Cardoso} J.~F.,  {Villaescusa-Navarro} F.,  {Ho} S.,   {Mallat} S.,  2020, \mn@doi [Physical Review D] {10.1103/PhysRevD.102.103506}, \href {https://ui.adsabs.harvard.edu/abs/2020PhRvD.102j3506A} {102, 103506}

\bibitem[\protect\citeauthoryear{{Auclair} et~al.,}{{Auclair} et~al.}{2024}]{auclair2023separation}
{Auclair} C.,  et~al., 2024, \mn@doi [Astronomy \& Astrophysics] {10.1051/0004-6361/202346814}, \href {https://ui.adsabs.harvard.edu/abs/2024A&A...681A...1A} {681, A1}

\bibitem[\protect\citeauthoryear{Bak, Tang  \& Wiesenfeld}{Bak et~al.}{1987}]{Bak_1987}
Bak P.,  Tang C.,   Wiesenfeld K.,  1987, \mn@doi [Phys. Rev. Lett.] {10.1103/PhysRevLett.59.381}, 59, 381

\bibitem[\protect\citeauthoryear{Bickel \& Levina}{Bickel \& Levina}{2008}]{Bickel_2008}
Bickel P.~J.,  Levina E.,  2008, \mn@doi [The Annals of Statistics] {10.1214/08-AOS600}, 36, 2577

\bibitem[\protect\citeauthoryear{Bouchaud \& Georges}{Bouchaud \& Georges}{1990}]{Bouchaud_1990}
Bouchaud J.-P.,  Georges A.,  1990, \mn@doi [Physics Reports] {https://doi.org/10.1016/0370-1573(90)90099-N}, 195, 127

\bibitem[\protect\citeauthoryear{{Bougeret} et~al.,}{{Bougeret} et~al.}{1995}]{Bougeret_1995}
{Bougeret} J.~L.,  et~al., 1995, \mn@doi ["Space Science Reviews"] {10.1007/BF00751331}, \href {https://ui.adsabs.harvard.edu/abs/1995SSRv...71..231B} {71, 231}

\bibitem[\protect\citeauthoryear{Brillinger}{Brillinger}{1965}]{brillinger1965introduction}
Brillinger D.~R.,  1965, The Annals of mathematical statistics, pp 1351--1374

\bibitem[\protect\citeauthoryear{Bruna \& Mallat}{Bruna \& Mallat}{2013}]{bruna2013invariant}
Bruna J.,  Mallat S.,  2013, IEEE transactions on pattern analysis and machine intelligence, 35, 1872

\bibitem[\protect\citeauthoryear{Bruna \& Mallat}{Bruna \& Mallat}{2019}]{bruna2019multiscale}
Bruna J.,  Mallat S.,  2019, Mathematical Statistics and Learning, 1, 257

\bibitem[\protect\citeauthoryear{Cai \& Liu}{Cai \& Liu}{2011}]{Cai_2011}
Cai T.,  Liu W.,  2011, \mn@doi [Journal of the American Statistical Association] {10.1198/jasa.2011.tm10560}, 106, 672

\bibitem[\protect\citeauthoryear{Chang, Yu  \& Vetterli}{Chang et~al.}{2000}]{chang2000adaptive}
Chang S.~G.,  Yu B.,   Vetterli M.,  2000, IEEE transactions on image processing, 9, 1532

\bibitem[\protect\citeauthoryear{Cheng \& M{\'e}nard}{Cheng \& M{\'e}nard}{2021a}]{cheng2021quantify}
Cheng S.,  M{\'e}nard B.,  2021a, arXiv preprint arXiv:2112.01288

\bibitem[\protect\citeauthoryear{{Cheng} \& {M{\'e}nard}}{{Cheng} \& {M{\'e}nard}}{2021b}]{Cheng_2021}
{Cheng} S.,  {M{\'e}nard} B.,  2021b, \mn@doi [Monthly Notices of the Royal Astronomical Society] {10.1093/mnras/stab2102}, \href {https://ui.adsabs.harvard.edu/abs/2021MNRAS.507.1012C} {507, 1012}

\bibitem[\protect\citeauthoryear{{Cheng}, {Ting}, {M{\'e}nard}  \& {Bruna}}{{Cheng} et~al.}{2020}]{Cheng_2020}
{Cheng} S.,  {Ting} Y.-S.,  {M{\'e}nard} B.,   {Bruna} J.,  2020, \mn@doi ["Monthly Notices of the Royal Astronomical Society "] {10.1093/mnras/staa3165}, \href {https://ui.adsabs.harvard.edu/abs/2020MNRAS.499.5902C} {499, 5902}

\bibitem[\protect\citeauthoryear{{Coles} \& {Jones}}{{Coles} \& {Jones}}{1991}]{Coles_1991}
{Coles} P.,  {Jones} B.,  1991, \mn@doi ["Monthly Notices of the Royal Astronomical Society "] {10.1093/mnras/248.1.1}, \href {https://ui.adsabs.harvard.edu/abs/1991MNRAS.248....1C} {248, 1}

\bibitem[\protect\citeauthoryear{Delouis, Allys, Gauvrit  \& Boulanger}{Delouis et~al.}{2022}]{delouis2022non}
Delouis J.-M.,  Allys E.,  Gauvrit E.,   Boulanger F.,  2022, Astronomy \& Astrophysics, 668, A122

\bibitem[\protect\citeauthoryear{Donoho \& Johnstone}{Donoho \& Johnstone}{1994}]{Donoho_1994}
Donoho D.~L.,  Johnstone I.~M.,  1994, \mn@doi [Biometrika] {10.1093/biomet/81.3.425}, 81, 425

\bibitem[\protect\citeauthoryear{Dudok~de Wit}{Dudok~de Wit}{2004}]{Dudok_2004}
Dudok~de Wit T.,  2004, \mn@doi [Phys. Rev. E] {10.1103/PhysRevE.70.055302}, 70, 055302

\bibitem[\protect\citeauthoryear{{Eddington}}{{Eddington}}{1914}]{Eddington_1914}
{Eddington} A.~S.,  1914, {Stellar movements and the structure of the universe}.
London, Macmillan and co., limited

\bibitem[\protect\citeauthoryear{Fan, Liao  \& Mincheva}{Fan et~al.}{2013}]{Fan_2013}
Fan J.,  Liao Y.,   Mincheva M.,  2013, Journal of the Royal Statistical Society. Series B, Statistical methodology, 75

\bibitem[\protect\citeauthoryear{Geman \& Geman}{Geman \& Geman}{1984}]{German_1984}
Geman S.,  Geman D.,  1984, \mn@doi [IEEE Transactions on Pattern Analysis and Machine Intelligence] {10.1109/TPAMI.1984.4767596}, PAMI-6, 721

\bibitem[\protect\citeauthoryear{{Greig}, {Ting}  \& {Kaurov}}{{Greig} et~al.}{2022}]{Greig_2022}
{Greig} B.,  {Ting} Y.-S.,   {Kaurov} A.~A.,  2022, \mn@doi ["Monthly Notices of the Royal Astronomical Society "] {10.1093/mnras/stac977}, \href {https://ui.adsabs.harvard.edu/abs/2022MNRAS.513.1719G} {513, 1719}

\bibitem[\protect\citeauthoryear{Gupta, Matilla, Hsu  \& Haiman}{Gupta et~al.}{2018}]{gupta2018non}
Gupta A.,  Matilla J. M.~Z.,  Hsu D.,   Haiman Z.,  2018, Physical Review D, 97, 103515

\bibitem[\protect\citeauthoryear{Ha, Singh, Lanusse, Upadhyayula  \& Yu}{Ha et~al.}{2021}]{ha2021adaptive}
Ha W.,  Singh C.,  Lanusse F.,  Upadhyayula S.,   Yu B.,  2021, Advances in Neural Information Processing Systems, 34, 20669

\bibitem[\protect\citeauthoryear{Huber}{Huber}{1981}]{Huber_1981}
Huber P.~J.,  1981, Robust statistics.
Wiley Ser. Probab. Math. Stat., John Wiley \& Sons, Hoboken, NJ, \mn@doi{10.1002/0471725250}

\bibitem[\protect\citeauthoryear{Jaffard}{Jaffard}{2004}]{jaffard2004wavelet}
Jaffard S.,  2004, Technical report, Wavelet techniques in multifractal analysis.
PARIS UNIV (FRANCE)

\bibitem[\protect\citeauthoryear{{Jaynes}}{{Jaynes}}{1957}]{jaynes_1957}
{Jaynes} E.~T.,  1957, \mn@doi [Physical Review] {10.1103/PhysRev.106.620}, \href {https://ui.adsabs.harvard.edu/abs/1957PhRv..106..620J} {106, 620}

\bibitem[\protect\citeauthoryear{{Jeffrey}, {Boulanger}, {Wandelt}, {Regaldo-Saint Blancard}, {Allys}  \& {Levrier}}{{Jeffrey} et~al.}{2022}]{jeffrey2022single}
{Jeffrey} N.,  {Boulanger} F.,  {Wandelt} B.~D.,  {Regaldo-Saint Blancard} B.,  {Allys} E.,   {Levrier} F.,  2022, \mn@doi ["Monthly Notices of the Royal Astronomical Society "] {10.1093/mnrasl/slab120}, \href {https://ui.adsabs.harvard.edu/abs/2022MNRAS.510L...1J} {510, L1}

\bibitem[\protect\citeauthoryear{Kello, Brown, i Cancho, Holden, Linkenkaer-Hansen, Rhodes  \& {Van Orden}}{Kello et~al.}{2010}]{Kello_2010}
Kello C.~T.,  Brown G.~D.,  i Cancho R.~F.,  Holden J.~G.,  Linkenkaer-Hansen K.,  Rhodes T.,   {Van Orden} G.~C.,  2010, \mn@doi [Trends in Cognitive Sciences] {https://doi.org/10.1016/j.tics.2010.02.005}, 14, 223

\bibitem[\protect\citeauthoryear{Kolmogorov}{Kolmogorov}{1941}]{kolmogorov1941local}
Kolmogorov A.~N.,  1941, Cr Acad. Sci. URSS, 30, 301

\bibitem[\protect\citeauthoryear{Landau \& Lifshitz}{Landau \& Lifshitz}{2013}]{landau2013statistical}
Landau L.~D.,  Lifshitz E.~M.,  2013, Statistical Physics: Volume 5.
~ Vol. 5, Elsevier

\bibitem[\protect\citeauthoryear{{Li} et~al.,}{{Li} et~al.}{2008}]{Li_2008}
{Li} Y.,  et~al., 2008, \mn@doi [Journal of Turbulence] {10.1080/14685240802376389}, \href {https://ui.adsabs.harvard.edu/abs/2008JTurb...9...31L} {9, N31}

\bibitem[\protect\citeauthoryear{Lombardo, Volpi, Koutsoyiannis  \& Papalexiou}{Lombardo et~al.}{2014}]{Lombardo_2014}
Lombardo F.,  Volpi E.,  Koutsoyiannis D.,   Papalexiou S.~M.,  2014, \mn@doi [Hydrology and Earth System Sciences] {10.5194/hess-18-243-2014}, 18, 243

\bibitem[\protect\citeauthoryear{Mallat}{Mallat}{1999}]{stephane1999wavelet}
Mallat S.,  1999, A wavelet tour of signal processing

\bibitem[\protect\citeauthoryear{{Mallat}}{{Mallat}}{2012}]{mallatscat}
{Mallat} S.,  2012, \mn@doi [Communications on Pure and Applied Mathematics] {10.1002/cpa.21413}, \href {https://arxiv.org/abs/1101.2286} {65, 1331}

\bibitem[\protect\citeauthoryear{Mallat, Zhang  \& Rochette}{Mallat et~al.}{2020}]{mallat2020phase}
Mallat S.,  Zhang S.,   Rochette G.,  2020, \mn@doi [Information and Inference: A Journal of the IMA] {10.1093/imaiai/iaz019}, 9, 721

\bibitem[\protect\citeauthoryear{Mandelbrot, Fisher  \& Calvet}{Mandelbrot et~al.}{1997}]{mandelbrot1997multifractal}
Mandelbrot B.~B.,  Fisher A.~J.,   Calvet L.~E.,  1997, Cowles Foundation discussion paper, p.~1412

\bibitem[\protect\citeauthoryear{Marchand, Ozawa, Biroli  \& Mallat}{Marchand et~al.}{2023}]{WCRG}
Marchand T.,  Ozawa M.,  Biroli G.,   Mallat S.,  2023, \mn@doi [Phys. Rev. X] {10.1103/PhysRevX.13.041038}, \href {https://ui.adsabs.harvard.edu/abs/2022arXiv220704941M} {13, 041038}

\bibitem[\protect\citeauthoryear{Matilla, Haiman, Hsu, Gupta  \& Petri}{Matilla et~al.}{2016}]{matilla2016dark}
Matilla J. M.~Z.,  Haiman Z.,  Hsu D.,  Gupta A.,   Petri A.,  2016, Physical Review D, 94, 083506

\bibitem[\protect\citeauthoryear{Meyer}{Meyer}{1992}]{Meyer:92c}
Meyer Y.,  1992, Wavelets and Operators.
Advanced mathematics. Cambridge university press

\bibitem[\protect\citeauthoryear{Morel, Rochette, Leonarduzzi, Bouchaud  \& Mallat}{Morel et~al.}{2022}]{rudy}
Morel R.,  Rochette G.,  Leonarduzzi R.,  Bouchaud J.-P.,   Mallat S.,  2022, arXiv preprint arXiv:2204.10177

\bibitem[\protect\citeauthoryear{{Neyrinck}, {Szapudi}  \& {Szalay}}{{Neyrinck} et~al.}{2009}]{Neyrinck_2009}
{Neyrinck} M.~C.,  {Szapudi} I.,   {Szalay} A.~S.,  2009, \mn@doi ["Astrophysical Journal Letter"] {10.1088/0004-637X/698/2/L90}, \href {https://ui.adsabs.harvard.edu/abs/2009ApJ...698L..90N} {698, L90}

\bibitem[\protect\citeauthoryear{Olah, Mordvintsev  \& Schubert}{Olah et~al.}{2017}]{olah2017feature}
Olah C.,  Mordvintsev A.,   Schubert L.,  2017, \mn@doi [Distill] {10.23915/distill.00007}

\bibitem[\protect\citeauthoryear{Olshausen \& Field}{Olshausen \& Field}{1996}]{olshausen_1996}
Olshausen B.~A.,  Field D.~J.,  1996, \mn@doi [Nature] {10.1038/381607a0}, 381, 607

\bibitem[\protect\citeauthoryear{Perlman, Burns, Li  \& Meneveau}{Perlman et~al.}{2007}]{Perlman_2007}
Perlman E.,  Burns R.,  Li Y.,   Meneveau C.,  2007, in Proceedings of the 2007 ACM/IEEE Conference on Supercomputing. SC '07.
Association for Computing Machinery, New York, NY, USA, \mn@doi{10.1145/1362622.1362654}, \url {https://doi.org/10.1145/1362622.1362654}

\bibitem[\protect\citeauthoryear{{Podesta}}{{Podesta}}{2009}]{Podesta_2009}
{Podesta} J.~J.,  2009, \mn@doi ["Astrophysical Journal"] {10.1088/0004-637X/698/2/986}, \href {https://ui.adsabs.harvard.edu/abs/2009ApJ...698..986P} {698, 986}

\bibitem[\protect\citeauthoryear{Portilla \& Simoncelli}{Portilla \& Simoncelli}{2000}]{portilla2000parametric}
Portilla J.,  Simoncelli E.~P.,  2000, International journal of computer vision, 40, 49

\bibitem[\protect\citeauthoryear{Regaldo-Saint~Blancard, Allys, Boulanger, Levrier  \& Jeffrey}{Regaldo-Saint~Blancard et~al.}{2021}]{regaldo2021new}
Regaldo-Saint~Blancard B.,  Allys E.,  Boulanger F.,  Levrier F.,   Jeffrey N.,  2021, Astronomy \& Astrophysics, 649, L18

\bibitem[\protect\citeauthoryear{R{\'e}galdo-Saint~Blancard, Allys, Auclair, Boulanger, Eickenberg, Levrier, Vacher  \& Zhang}{R{\'e}galdo-Saint~Blancard et~al.}{2023}]{regaldo2023generative}
R{\'e}galdo-Saint~Blancard B.,  Allys E.,  Auclair C.,  Boulanger F.,  Eickenberg M.,  Levrier F.,  Vacher L.,   Zhang S.,  2023, The Astrophysical Journal, 943, 9

\bibitem[\protect\citeauthoryear{Saydjari, Portillo, Slepian, Kahraman, Burkhart  \& Finkbeiner}{Saydjari et~al.}{2021}]{saydjari2021classification}
Saydjari A.~K.,  Portillo S.~K.,  Slepian Z.,  Kahraman S.,  Burkhart B.,   Finkbeiner D.~P.,  2021, The Astrophysical Journal, 910, 122

\bibitem[\protect\citeauthoryear{{Schneider}, {Farge}, {Azzalini}  \& {Ziuber}}{{Schneider} et~al.}{2006}]{schneider2006coherent}
{Schneider} K.,  {Farge} M.,  {Azzalini} A.,   {Ziuber} J.,  2006, \mn@doi [Journal of Turbulence] {10.1080/14685240600601061}, \href {https://ui.adsabs.harvard.edu/abs/2006JTurb...7...44S} {7, 44}

\bibitem[\protect\citeauthoryear{Sherman}{Sherman}{2018}]{Sherman_1996}
Sherman M.,  2018, \mn@doi [Journal of the Royal Statistical Society: Series B (Methodological)] {10.1111/j.2517-6161.1996.tb02097.x}, 58, 509

\bibitem[\protect\citeauthoryear{Siahkoohi, Morel, de Hoop, Allys, Sainton  \& Kawamura}{Siahkoohi et~al.}{2023}]{siahkoohi2023unearthing}
Siahkoohi A.,  Morel R.,  de Hoop M.~V.,  Allys E.,  Sainton G.,   Kawamura T.,  2023, arXiv preprint arXiv:2301.11981

\bibitem[\protect\citeauthoryear{Sornette}{Sornette}{2017}]{Sornette_2017}
Sornette D.,  2017, Why Stock Markets Crash: Critical Events in Complex Financial Systems, rev - revised edn.
Princeton University Press, \url {http://www.jstor.org/stable/j.ctt1h1htkg}

\bibitem[\protect\citeauthoryear{Stein}{Stein}{1956}]{Stein_1956}
Stein C.,  1956, in Proceedings of the {T}hird {B}erkeley {S}ymposium on {M}athematical {S}tatistics and {P}robability, 1954--1955, vol. {I}. University of California Press, Berkeley-Los Angeles, Calif., pp 197--206

\bibitem[\protect\citeauthoryear{Valogiannis \& Dvorkin}{Valogiannis \& Dvorkin}{2022a}]{valogiannis2022towards}
Valogiannis G.,  Dvorkin C.,  2022a, Physical Review D, 105, 103534

\bibitem[\protect\citeauthoryear{{Valogiannis} \& {Dvorkin}}{{Valogiannis} \& {Dvorkin}}{2022b}]{Valogiannis_2022}
{Valogiannis} G.,  {Dvorkin} C.,  2022b, \mn@doi ["Physical Review D"] {10.1103/PhysRevD.106.103509}, \href {https://ui.adsabs.harvard.edu/abs/2022PhRvD.106j3509V} {106, 103509}

\bibitem[\protect\citeauthoryear{{Valogiannis}, {Yuan}  \& {Dvorkin}}{{Valogiannis} et~al.}{2023}]{Valogiannis_2023}
{Valogiannis} G.,  {Yuan} S.,   {Dvorkin} C.,  2023, \mn@doi [arXiv e-prints] {10.48550/arXiv.2310.16116}, \href {https://ui.adsabs.harvard.edu/abs/2023arXiv231016116V} {p. arXiv:2310.16116}

\bibitem[\protect\citeauthoryear{{Vielva}, {Mart{\'\i}nez-Gonz{\'a}lez}, {Barreiro}, {Sanz}  \& {Cay{\'o}n}}{{Vielva} et~al.}{2004}]{Vielva_2004}
{Vielva} P.,  {Mart{\'\i}nez-Gonz{\'a}lez} E.,  {Barreiro} R.~B.,  {Sanz} J.~L.,   {Cay{\'o}n} L.,  2004, \mn@doi ["Astrophysical Journal"] {10.1086/421007}, \href {https://ui.adsabs.harvard.edu/abs/2004ApJ...609...22V} {609, 22}

\bibitem[\protect\citeauthoryear{{Villaescusa-Navarro} et~al.,}{{Villaescusa-Navarro} et~al.}{2020}]{quijote2020}
{Villaescusa-Navarro} F.,  et~al., 2020, \mn@doi ["Astrophysical Journal Supplement Series"] {10.3847/1538-4365/ab9d82}, \href {https://ui.adsabs.harvard.edu/abs/2020ApJS..250....2V} {250, 2}

\bibitem[\protect\citeauthoryear{Zhang \& Mallat}{Zhang \& Mallat}{2021}]{zhang2021maximum}
Zhang S.,  Mallat S.,  2021, Applied and Computational Harmonic Analysis, 53, 199

\bibitem[\protect\citeauthoryear{Zhu, Wu  \& Mumford}{Zhu et~al.}{1997}]{zhu1997minimax}
Zhu S.~C.,  Wu Y.~N.,   Mumford D.,  1997, Neural computation, 9, 1627

\bibitem[\protect\citeauthoryear{Zhu, Wu  \& Mumford}{Zhu et~al.}{1998}]{zhu1998filters}
Zhu S.~C.,  Wu Y.,   Mumford D.,  1998, International Journal of Computer Vision, 27, 107

\makeatother
\end{thebibliography}

\end{document}